\documentclass[11pt,a4paper]{article}

\parskip=\baselineskip

\usepackage{tikz}
\usetikzlibrary{arrows,decorations.markings,shapes.arrows,patterns, decorations.pathmorphing}
\usepackage{amsmath,amsthm} 
\usepackage{color,verbatim}
\usepackage[colorlinks=true,linkcolor=blue,citecolor=red]{hyperref}
\usepackage{amssymb}

\tikzset{
  ->-/.style={decoration={markings, mark=at position 0.5 with {\arrow{to}}},
              postaction={decorate}},
}
\tikzset{-<-/.style={decoration={markings,mark=at position 0.5 with %
    {\arrow[scale=1.5,>=stealth]{<}}},postaction={decorate}}}

\numberwithin{equation}{section}
\newcommand{\mc}{\mathcal}
\newcommand{\mf}{\mathfrak}
\renewcommand{\sl}{\mathfrak{sl}}
\def\tilde{\widetilde}
\def\bar{\overline}
\newcommand{\so}{\mathfrak{so}}
\newcommand{\eps}{\epsilon}
\newcommand{\g}{\mathfrak{g}}
\newcommand{\tr}{\triangle}

\newcommand{\til}{\widetilde}

\newcommand{\C}{\mathbb C}

\newcommand{\Z}{\mathbb Z}

\newcommand{\op}{\operatorname}
\newcommand{\mbf}{\mathbf}
\newcommand{\mbb}{\mathbb}

\newcommand{\ip}[1]{\left\langle #1 \right\rangle}

\newcommand{\dtil}[1]{\bar{#1}}

\newcommand{\R}{\mbb R}
\renewcommand{\d}{\mathrm{d}}

\def\Tr{{\mathrm{Tr}}}
\def\O{{\mathcal O}}
\def\sh{{\sf h}}
\def\i{\mathsf{i}}

\def\be{\begin{equation}}
\def\ee{\end{equation}}

\def\gl{\mathfrak{gl}}
\def\sl{\mathfrak{sl}}
\def\so{\mathfrak{so}}
\def\sp{\mathfrak{sp}}
\def\gtwo{\mathfrak{g}_2}
\def\ffour{\mathfrak{f}_4}
\def\esix{\mathfrak{e}_6}
\def\eseven{\mathfrak{e}_7}
\def\eeight{\mathfrak{e}_8}

\newtheorem{theorem}{Theorem}[section]
\newtheorem{proposition}[theorem]{Proposition}

\newtheorem{definition}[theorem]{Definition}

\def\eqn#1{eqn.\ #1}
\def\eqns#1{eqns.\ #1}
\def\fig#1{Fig.\ #1}

\begin{document}

\pagenumbering{Alph} 
\begin{titlepage}
\begin{flushright}

\end{flushright}
\vskip 1.5in
\begin{center}
{\bf\Large{Gauge Theory And Integrability, II}}
\vskip 0.5cm 
{Kevin Costello$^1$, Edward Witten$^2$ and Masahito Yamazaki$^3$} \vskip 0.05in 
{\small{ 
\textit{$^1$Perimeter Institute for Theoretical Physics, }\vskip -.4cm
\textit{Waterloo, ON N2L 2Y5, Canada}
\vskip 0 cm 
\textit{$^2$School of Natural Sciences, Institute for Advanced Study,}\vskip -.4cm
\textit{Einstein Drive, Princeton, NJ 08540 USA}
\vskip 0 cm 
\textit{$^3$Kavli Institute for the Physics and Mathematics of the Universe (WPI), }\vskip -.4cm
\textit{University of Tokyo, Kashiwa, Chiba 277-8583, Japan}
}}

\end{center}

\vskip 0.5in
\baselineskip 16pt
\begin{abstract}
Starting with a four-dimensional gauge theory approach  to rational, elliptic, and trigonometric solutions of the Yang-Baxter equation,  we determine the corresponding
quantum group deformations  to all orders in $\hbar$ by deducing their RTT presentations.  The arguments we give are a mix of familiar
ones with reasoning that is more transparent from the four-dimensional gauge theory point of view.  The arguments apply most directly for $\gl_N$ and can be extended to
all simple Lie algebras other than $\eeight$ by taking into account the self-duality of some representations, the framing anomaly for Wilson operators, and the existence
of quantum vertices at which several Wilson operators can end.
\end{abstract}
\date{September, 2017}
\end{titlepage}
\pagenumbering{arabic} 

\tableofcontents

\section{Introduction}

Several years ago, it was argued by one of us \cite{Costello_2013}, originally on the basis of relatively abstract arguments, that the usual rational, elliptic, and trigonometric
solutions of the Yang-Baxter equation can be systematically derived from a certain four-dimensional gauge theory.   Recently \cite{Part1}, we have reformulated this approach in a 
more direct way and developed it further.  Familiarity with that paper will be assumed here.  For an informal introduction to this subject, see \cite{Witten_2016}.

The four-dimensional gauge theory in question is most simply defined on a  product four-manifold $\Sigma\times C$, where $\Sigma$ is a smooth oriented two-manifold
and $C$ is a complex Riemann surface.  In the present paper, our considerations are local along $\Sigma$, so the choice of $\Sigma$ does not matter.  We will simply
take $\Sigma$ to be the $xy$ plane.  $C$ on the other hand is endowed with a holomorphic differential $\omega$ that has no zeroes  and that has poles at infinity along $C$
(that is, there is a compactification $\bar C$ of $C$ such that $\omega$ has a pole at each point of $\bar C\backslash C$).
  This condition leaves three choices of $C$, which may be the complex plane $\C$, $\C^\times =\C\backslash \{0\}$, or a Riemann surface
of genus 1.   These choices correspond to the rational, elliptic, and trigonometric solutions of the Yang-Baxter equation.

In \cite{Part1}, the relevant structures were described and computed to lowest nontrivial order in the quantum deformation parameter $\hbar$.  If $G$ is the gauge group and 
$\g$ its Lie algebra, then at the classical level the theory has Wilson line operators associated to representations of the algebra $\g[[z]]$. At the quantum level, there is
a nontrivial $R$-matrix and a nontrivial operator product expansion (OPE) for these line operators; in addition they are subject to a framing anomaly.  Moreover, at the quantum level
the algebra $\g[[z]]$, or more precisely its universal enveloping algebra, is deformed to a quantum group.    

In the rational,  trigonometric, or elliptic case, the relevant quantum
group is known as the Yangian, the quantum loop group, or the elliptic quantum group.  However, in \cite{Part1}, we saw the deformation from $\g[[z]]$ only to lowest nontrivial
order in $\hbar$.  In the present paper, we will describe a way
to extend the analysis to all orders in $\hbar$.   

We should clarify the meaning of the phrase ``all orders.''  Many of the formulas in this paper make sense and are valid with $\hbar$ treated as a complex number.  However,
the theory introduced in \cite{Costello_2013} and further developed here is, in its present form, only valid perturbatively.  Hence, statements in the present paper really refer
to formal power series in $\hbar$.

To find a picture valid to all orders, we will describe in this paper what are known in the literature as RTT presentations for the relevant quantum groups
 \cite{Takhtajan-Faddeev,Kulish-Sklyanin,Drinfeld_Hopf,Drinfeld_original,FRT_1988,Chari-Pressley}.  
 We describe RTT presentations in the simplest case,  the Yangian of $\gl_N$, in section \ref{RTT_gl}.  The arguments are based on simple manipulations of Wilson operators.
  In sections \ref{RTT_so_sp}, \ref{RTT_sl} , and \ref{RTTexc}, we extend this treatment
successively to $\so_N$, $\sp_{2N}$, $\sl_N$, and to all exceptional simple Lie algebras other than\footnote{Unfortunately, our approach does not work conveniently
for $\eeight$ because its Yangian algebra does not have a convenient representation.  See 
section \ref{RTTexc} for a fuller explanation.}  $\eeight$. This involves several new ingredients relative to the case of $\gl_N$, mainly the self-dual nature of certain Wilson
operators, the framing anomaly, and the existence of vertices on which several Wilson operators can end. In section \ref{unirational}, we explain in what sense the algebra described by RTT presentations, for any of these cases,
is unique as a deformation of the universal enveloping algebra of $\g[[z]]$.    In section \ref{trigonometric}, we extend the analysis to describe RTT presentations in the trigonometric case.
This is more complicated than the rational case, since
 trigonometric solutions of the Yang-Baxter equations have less symmetry than rational ones, and since in the gauge theory language a rather subtle boundary condition
is needed to describe the trigonometric solutions of Yang-Baxter.  It turns out that our analysis in section \ref{trigonometric} gives an interesting perspective on purely three-dimensional Chern-Simons theory.
We investigate in section \ref{unitrigonometric} the uniqueness of the algebras obtained from the trigonometric RTT presentations, and meet a small surprise, which however turns out to have a simple explanation in the gauge theory language: the most obvious uniqueness hypothesis is not valid, as there are additional deformation parameters that appear
when $\g$ has rank greater than 1. Finally, in  section \ref{elliptic}, we explain what one can say along similar lines in the elliptic case.   

Many results in this paper are known from other points of view.  For example, the RTT presentations which are our main focus are certainly already known, at least for classical groups.
The arguments in this paper are a mix of familiar ones with reasoning that is more transparent from the four-dimensional gauge theory point of view.  Some
 results presented here may be novel and some arguments may add something to what is previously known.  For example, the quantum determinant for $\sl_N$ is certainly
already known, but it possibly adds something to derive it from a vertex with manifest symmetry that can be constructed by elementary arguments.  Also possibly novel
are the  analogs we construct of the quantum
determinant for
exceptional algebras (other than $\eeight$) and the resulting RTT presentations, as well as
 the additional deformation parameters that we find in section \ref{trigonometric} in the trigonometric case for algebras of
rank greater than 1.

\section{\texorpdfstring{RTT Relation for $Y(\gl_N)$}{RTT Relation for Y(gl(n))}}\label{RTT_gl}

\subsection{Basics}\label{basics}

We begin with the simplest case: the Yangian for the algebra $\mathfrak{g}=\gl_N$. This corresponds to taking $C$ to be the complex $z$-plane $\C$.

Consider a general Wilson line supported at $z = 0$, and running along a straight line in the $xy$ plane which we will draw as horizontal.  
To specify such a Wilson line at the classical level, we give a finite-dimensional vector space $W$, together with a sequence of operators
\begin{equation}
t_{a,n,W} : W \to W
\end{equation}
for $n \ge 0$.  We assume that $t_{a,n,W} = 0$ for $n \gg 0$.  As we saw in \cite{Part1}, section 3.3,   from such a sequence of operators
we can construct a classical Wilson line in which the $n^{\rm th}$ derivative of the gauge field $\frac{1}{n!}\partial_z^n A_{x,a}$ is coupled by $t_{a,n,W}$.  At the classical level, the Wilson line is gauge-invariant as long as the commutation relations
\begin{equation}\label{logo}
[t_{a,n,W},t_{b,m,W}]=f_{ab}{}^c t_{c,n+m,W} 
\end{equation}
are satisfied.  

The situation is different at the quantum level.
In order for the Wilson line operator to be gauge-invariant at the quantum level, the operators $t_{a,n,W}$ need to satisfy a deformed version of the relation \eqref{logo}. 
In our previous paper \cite{Part1} we derived this fact using two different methods;
in section 5.4 (and in particular \eqn (5.23)) we derived the order $\hbar^2$ correction to the relation by a slightly indirect method,
and in section 8 by an explicit two-loop computation.

Here, we will derive the quantum corrections to this relation in yet another way, which will produce an answer that is valid in all orders in $\hbar$.  The algebra generated by $t_{a,n}$ and satisfying these quantum-corrected commutation relations is known as the Yangian algebra.    The analysis is most simple for $\mf{g} = \mf{gl}_N$,  and we will make this assumption throughout the present section. 

Let $e^i_j$ be the elementary $N\times N$ matrix, with 1 in the $(i,j)$ entry and 0 elsewhere.   We write $t^i_j$ for the corresponding generator of $\gl_N$.  In the basis of $\gl_N$ given by the $t^i_j$, the Lie
algebra commutation relations read 
\begin{align}
[t^i_j, t^k_l]=\delta_j^k t^i_l - \delta^i_l t^k_j \;,
\label{gl_commutation}
\end{align}
and the invariant bilinear form is
\be\label{inform} 
(t^i_j, t^k_l)=\delta^i_l \delta^k_j  \;.
\ee
Since the notation
$t^i_{j,n}$ for the generator of $\gl_N[[z]]$ corresponding to $t^i_jz^n$ seems clumsy, we will write instead $t^i_j[n]$.   Acting on a tensor product $W_1\otimes W_2$ of two
representations of $\gl_N$, the invariant bilinear form (\ref{inform}) determines an invariant operator which in the basis given by the $t^i_j$ is just
\be\label{cba}
c=\sum_{i,j} t^i_j\otimes t^j_i \;. 
\ee

Now, consider a pair of Wilson lines in the $xy$ plane.  In the horizontal direction, we take a general Wilson operator associated
to a representation $W$.  For the moment, we assume that $W$ is a representation of $\gl_N$, and not a more general representation of $\gl_N[[z]]$.  Thus classically
this Wilson line couples only to the gauge field $A$ and not to its $z$ derivatives.  In addition, we take a vertical Wilson line in the fundamental, $N$-dimensional representation
of $\gl_N$, at an arbitrary value of $z$.   
The quasiclassical $r$-matrix, appearing in $R=1+\hbar r+\dots$, is in general $c/z$, as computed in \cite{Part1}, section 4, where $c$ for $\gl_N$ is defined in eqn.
(\ref{cba}).  In the present context, this is 
\begin{equation}\label{presc}
r=\sum_{i,j} \frac{1}{z} t^i_{j,W} \otimes e^j_i : W \otimes \C^N \to W \otimes \C^N  \;.
\end{equation} 
(We use the fact that the generators of $\gl_N$ in the fundamental representation are $t^i_j=e^i_j$.)  

Now, suppose that the horizontal Wilson line $W$ is associated classically to an arbitrary representation of $\gl_N[[z]]$
so that $\frac{1}{k!}\partial_z^k A$ is coupled to an operator $t^i_{j,W}[k]$ (if no confusion arises we will just write $t^i_j[k]$).   Crossing this with the same vertical Wilson line as before, the same calculation as in \cite{Part1} tells us that, to lowest order in $\hbar$, the $r$-matrix is
\begin{equation}\label{genc}
r=\sum_{i,j} \sum_{k \ge 0} \frac{1}{k!} \partial_z^k \frac{1}{z} (t^i_{j,W}[k] \otimes e^j_i )
=\sum_{i,j} \sum_{k \ge 0}  (-1)^k \frac{1}{z^{k+1}} (t^i_{j,W}[k] \otimes e^j_i )
\;.
\end{equation}

This tells us that, to leading order in $\hbar$, we can recover the operators $t^i_{j,W}[k]$ acting on $W$ by the following procedure. We cross with a vertical Wilson line in the fundamental representation at some point $z\in C$.   We place incoming and outgoing states on $\langle i | $ and $| j \rangle$ on this vertical Wilson line above and below the
point at which it crosses the horizontal Wilson line.\footnote{As explained in
\cite{Part1}, because the theory is infrared-free, away from crossing points one can label each Wilson line by a specified state in the appropriate representation of 
$\g[[z]]$.}  The $R$-matrix specialized in this way gives an operator
\begin{equation}
T^i_j(z) : W \to W 
\end{equation}
as in \fig \ref{figure_crossing}. 
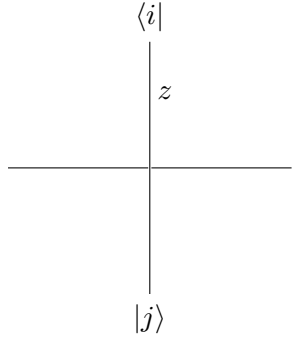
\begin{figure}[htbp]
\begin{center}
\begin{tikzpicture}
\node(In) at (0,4) {$\langle i|$};
\node(Out) at (0,0) {$|j \rangle$};
\node(HLeft) at (-2,2) {};
\node(HRight) at (2,2) {};
\draw (HLeft) to (HRight); 
\draw[very thick, color=white] (In) to (Out);
\draw (In) to (Out);  
\node at (0.2,3) {$z$};
\end{tikzpicture}
\caption{\small{A vertical Wilson line, equipped with an incoming state $\langle i |$ and an outgoing state $|j \rangle$, gives rise to an operator acting on the states of a horizontal Wilson line. }}
\label{figure_crossing}
\end{center}
\end{figure}
If we then expand  $T^i_j(z)$ in powers of $z$ we find
\begin{equation}\label{T_series}
T^i_j(z) = \delta^i_j + \frac{ \hbar}{z} t^i_j[0] -  \frac{ \hbar }{z^{2}} t^i_j[1] + \dots + (-1)^{k}\frac{ \hbar}{z^{k+1}} t^i_j[k] + \dots .
\end{equation}

Comparing to eqn. (\ref{genc}), we see that the operators $t^i_j[m]$ reduce, modulo $\hbar$, to the classical generators of $\gl_N[[z]]$.  We can, however, take eqn. (\ref{T_series})
as a {\it definition} of operators $t^i_j[m]:W\to W$ to all orders in $\hbar$. 
The goal of this section is to derive the relations satisfied by those operators at the \emph{quantum} level. We will do this by deriving a relation, known as the RTT relation, for the operator $T^i_j(z)$.  This will imply that the operators $t^i_j[k]$ satisfy the relations of the Yangian algebra, a deformation of the universal enveloping algebra of $\gl_N[[z]]$.

To deduce the RTT relations, we consider a more elaborate picture, with the same horizontal Wilson line as before but now a pair of vertical Wilson lines in the
fundamental representation of $\gl_N$, at points $z,z'\in C$.    This is illustrated in \fig \ref{figure_parallel}. 
We label the vertical Wilson lines with respective incoming and outgoing states $\langle i|$, $|j \rangle$ and $\langle k|$, $|l \rangle$,  as shown in the figure.
The operator we find on the horizontal Wilson line  (read from right to left in the figure) is then simply the  composition $T^i_j(z) T^k_l(z')$ of the separate
operators from the two vertical Wilson lines. 

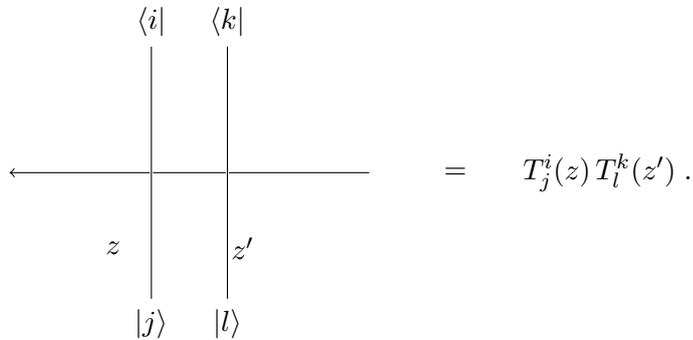
\begin{figure}[htbp]
\begin{center}
\begin{tikzpicture}
\node(In1) at (0,4) {$\langle i|$};
\node(Out1) at (0,0) {$|j \rangle$};
\node(In2) at (1,4) {$\langle k|$};
\node(Out2) at (1,0) {$|l \rangle$};
\node(HLeft) at (-2,2) {};
\node(HRight) at (3,2) {};
\draw[->] (HRight) to (HLeft); 
\draw[very thick, color=white] (In1) to (Out1);
\draw[very thick, color=white] (In2) to (Out2);
\draw[rounded corners] (In1.south) to  (Out1);
\node at (1.2,1) {$z'$};
\draw[rounded corners] (In2) to  (Out2);
\node at (-0.5,1) {$z$};
\node at (4,2) {$=$};
\node at (6,2) {$T^i_j(z) \, T^k_l(z') \;.$};
\end{tikzpicture}
\caption{\small{Two parallel vertical Wilson lines crossing a horizontal one lead to a composition of the corresponding operators.}}
\label{figure_parallel} 
\end{center}
\end{figure}

The goal now is to find relations among the operators $T^i_j(z)$ by moving the two vertical Wilson lines past each other.  However, if we do this while keeping them vertical,
then at some point they will coincide if projected to the $xy$ plane, and this does not lead to a conventional simple picture.  
It is more useful to let the ``vertical'' Wilson lines bend at some point, and cross each other, above or below the horizontal Wilson line.   There are two possible pictures,
as indicated in \fig \ref{figure_bent}, and they are equivalent by virtue of the same arguments that lead to the Yang-Baxter equation.

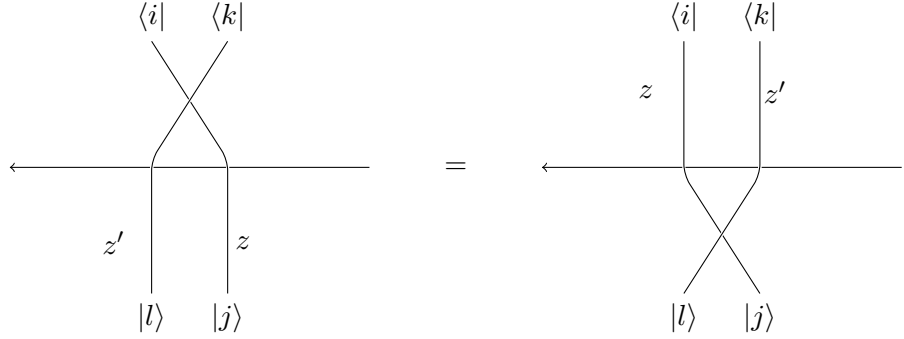
\begin{figure}[htbp]
\begin{center}
\begin{tikzpicture}
\begin{scope}
\node(In1) at (0,4) {$\langle i|$};
\node(Out1) at (1,0) {$|j \rangle$};
\node(In2) at (1,4) {$\langle k|$};
\node(Out2) at (0,0) {$|l \rangle$};
\node(HLeft) at (-2,2) {};
\node(HRight) at (3,2) {};
\draw[->] (HRight) to (HLeft); 
\draw[very thick, color=white] (In1) to (Out2);
\draw[very thick, color=white] (In2) to (Out1);
\draw[rounded corners] (In1.south) to (1,2.1) to  (Out1);
\node at (1.2,1) {$z$};
\draw[color=white, very thick, rounded corners] (In2.south) to (0,2.1) to  (Out2);
\draw[rounded corners] (In2.south) to (0,2.1) to  (Out2);
\node at (-0.5,1) {$z'$};
\end{scope}

\node at (4,2) {$=$};

\begin{scope}[shift={(7,0)}]
\node(In1) at (0,4) {$\langle i|$};
\node(Out1) at (1,0) {$|j \rangle$};
\node(In2) at (1,4) {$\langle k|$};
\node(Out2) at (0,0) {$|l \rangle$};
\node(HLeft) at (-2,2) {};
\node(HRight) at (3,2) {};
\draw[->] (HRight) to (HLeft);
\draw[very thick, color=white] (In1) to (Out2);
\draw[very thick, color=white] (In2) to (Out1);
\draw[rounded corners] (In1.south) to (0,1.9) to  (Out1.north);
\node at (-0.5,3) {$z$};
\draw[color=white, very thick, rounded corners] (In2.south) to (1,1.9) to  (Out2.north);
\draw[rounded corners] (In2.south) to (1,1.9) to  (Out2.north);
\node at (1.2,3) {$z'$};
\end{scope}

\end{tikzpicture}
\caption{\small{Two vertical Wilson lines are ``bent'' to cross each other, above or below a given horizontal Wilson line. (To avoid extraneous considerations involving
a framing anomaly, one can extend the ``vertical'' Wilson lines in this and subsequent pictures so that they are indeed asymptotically vertical.)}}
\label{figure_bent} 
\end{center}
\end{figure}

Let us now interpret both of these diagrams in terms of operators on the horizontal Wilson line $W$.  Let us consider the diagram on the left, and read it from top to bottom. At the crossing of the two vertical Wilson lines, the $R$-matrix $R(z-z')$ will act on the incoming and outgoing states.   The crossing of a vertical Wilson line with the horizontal
one gives a factor $T^i_j(z)$.   Thus, we can evaluate the diagrams as indicated in Figs. \ref{one_bent} and \ref{two_bent}.

\begin{figure}[htbp]
\begin{center}
\begin{tikzpicture}
\node(In1) at (0,4) {$\langle i|$};
\node(Out1) at (1,0) {$|j \rangle$};
\node(In2) at (1,4) {$\langle k|$};
\node(Out2) at (0,0) {$|l \rangle$};
\node(HLeft) at (-2,2) {};
\node(HRight) at (3,2) {};
\draw[->] (HRight) to (HLeft); 

\draw[very thick, color=white] (In1) to (Out2);
\draw[very thick, color=white] (In2) to (Out1);
\draw[rounded corners] (In1.south) to (1,2.1) to  (Out1);
\node at (1.2,1) {$z$};
\draw[color=white, very thick, rounded corners] (In2.south) to (0,2.1) to  (Out2);
\draw[rounded corners] (In2.south) to (0,2.1) to  (Out2);
\node at (-0.5,1) {$z'$};

\begin{scope}[shift={(9,0)}]
\node at (-3.7,2) {$= \displaystyle \sum_{r,s} R^{ik}_{rs} (z-z')$};  
\node(In1) at (0,4) {$\langle r|$};
\node(Out1) at (0,0) {$|j \rangle$};
\node(In2) at (1,4) {$\langle s|$};
\node(Out2) at (1,0) {$|l \rangle$};
\node(HLeft) at (-2,2) {};
\node(HRight) at (3,2) {};
\draw[->] (HRight) to (HLeft); 
\draw[very thick, color=white] (In1) to (Out1);
\draw[very thick, color=white] (In2) to (Out2);
\draw[rounded corners] (In1.south) to  (Out1);
\node at (1.2,1) {$z$};
\draw[rounded corners] (In2) to  (Out2);
\node at (-0.5,1) {$z'$};
\end{scope}
\end{tikzpicture}
\caption{\small{The first picture in \fig \ref{figure_bent} can be evaluated as shown here.}}\label{one_bent}
\end{center}
\end{figure}
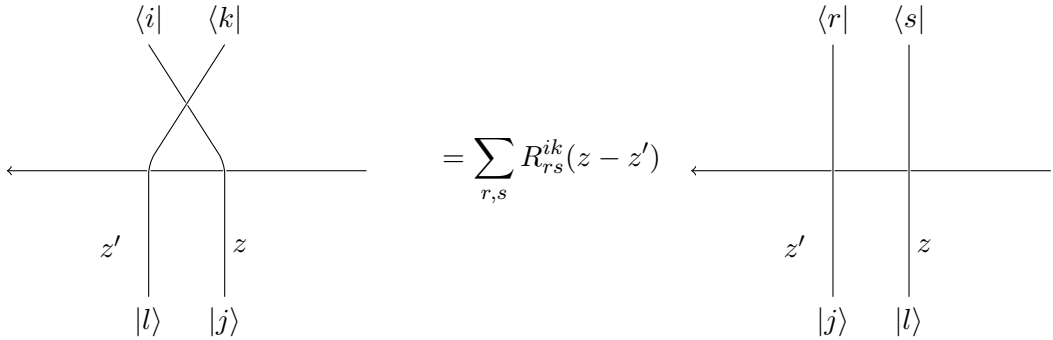

\begin{figure}[htbp]
\begin{center}
\begin{tikzpicture}
\node(In1) at (0,4) {$\langle i|$};
\node(Out1) at (1,0) {$|j \rangle$};
\node(In2) at (1,4) {$\langle k|$};
\node(Out2) at (0,0) {$|l \rangle$};
\node(HLeft) at (-2,2) {};
\node(HRight) at (3,2) {};
\draw[->] (HRight) to (HLeft); 
\draw[very thick, color=white] (In1) to (Out2);
\draw[very thick, color=white] (In2) to (Out1);
\draw[rounded corners] (In1.south) to (0,1.9) to  (Out1.north);
\node at (-0.5,3) {$z$};
\draw[color=white, very thick, rounded corners] (In2.south) to (1,1.9) to  (Out2.north);
\draw[rounded corners] (In2.south) to (1,1.9) to  (Out2.north);
\node at (1.2,3) {$z'$};

\begin{scope}[shift={(9,0)}]
\node at (-4,2) {$= \displaystyle \sum_{r,s} R^{rs}_{jl} (z-z')$};  
\node(In1) at (0,4) {$\langle i|$};
\node(Out1) at (0,0) {$|r \rangle$};
\node(In2) at (1,4) {$\langle k|$};
\node(Out2) at (1,0) {$|s \rangle$};
\node(HLeft) at (-2,2) {};
\node(HRight) at (3,2) {};
\draw[->] (HRight) to (HLeft);
\draw[very thick, color=white] (In1) to (Out1);
\draw[very thick, color=white] (In2) to (Out2);
\draw[rounded corners] (In1.south) to  (Out1);
\node at (1.2,3) {$z'$};
\draw[rounded corners] (In2) to  (Out2);
\node at (-0.5,3) {$z$};
\end{scope}
\end{tikzpicture}
\caption{\small{ The second picture in \fig \ref{figure_bent} has this interpretation.  }}\label{two_bent}
\end{center}
\end{figure}
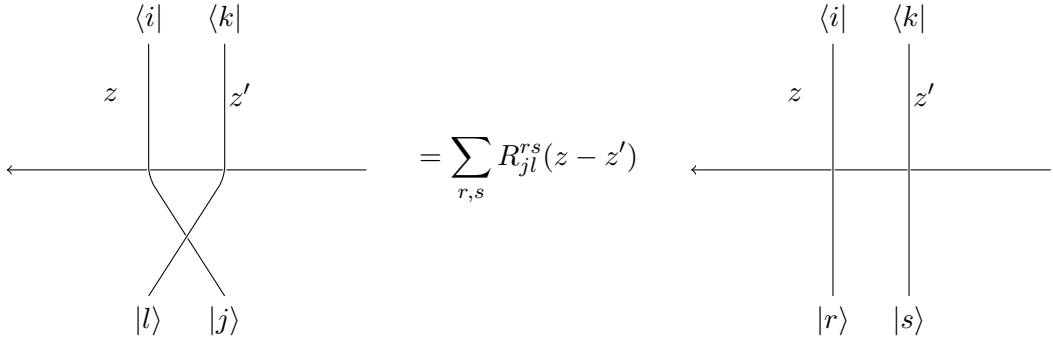

The equivalence between the two pictures gives us the following equation:
\begin{equation}
\sum_{r,s} R^{i k}_{r s}(z - z') T^r_j(z') T^s_l(z) = \sum_{r,s} T^i_r(z) T^k_s(z')R^{rs}_{jl}(z -z') \;.
\end{equation}
In interpreting this equation, note that the entries $R^{ij}_{rs}(z-z')$ of the $R$-matrix are scalar functions, whereas the entries $T^r_j$ of the $T$-operator are operators acting
on the representation $W$ associated to the horizontal Wilson line.

This equation is the \emph{RTT relation}, and is often written more succinctly as
\begin{equation}
R(z-z') T(z') T(z) = T(z) T(z') R(z-z') \;.
\label{RTT_glN}
\end{equation}
It is known (see \cite{Chari-Pressley,Drinfeld_original}) that this relation gives a presentation of the Yangian algebra for $\mf{gl}_N$.   We will explain in some detail
how this works.  

\subsection{\texorpdfstring{RTT relation in terms of an expansion of $T(z)$}{RTT relation in terms of an expansion of T(z)}}
As in \eqn \eqref{T_series}, let us expand the $z$-dependent operator on the horizontal Wilson line $T^i_j(z)$ in powers in $1/z$.
The coefficients  $t^i_j[m]$ in this expansion are a  quantum version of the operators which define the coupling of $\frac{1}{m!} \partial_z^m A$ to the Wilson line. 

Let $P : \C^N \otimes \C^N \to \C^N \otimes \C^N$ be the map which exchanges the two factors. The $R$ matrix for the fundamental representation of $\mf{gl}_N$ is\footnote{See for
example  \cite{Part1}, section 3.5. Note that an overall scalar factor multiplying  the $R$-matrix is irrelevant for our considerations here, as it does not affect the RTT relations.}
\begin{equation} 
R(z) =I +  \frac{\hbar}{z} P  \;.
 \end{equation}
This expression gives us the following diagrammatic identities:

\begin{center}
\begin{tikzpicture}
\node(In1) at (0,4) {$\langle i|$};
\node(Out1) at (1,0) {$|j \rangle$};
\node(In2) at (1,4) {$\langle k|$};
\node(Out2) at (0,0) {$|l \rangle$};
\node(HLeft) at (-2,2) {};
\node(HRight) at (3,2) {};
\draw[->] (HRight) to (HLeft);
\draw[very thick, color=white] (In1) to (Out2);
\draw[very thick, color=white] (In2) to (Out1);
\draw[rounded corners] (In1.south) to (1,2.1) to  (Out1);
\node at (1.2,1) {$z$};
\draw[color=white, very thick, rounded corners] (In2.south) to (0,2.1) to  (Out2);
\draw[rounded corners] (In2.south) to (0,2.1) to  (Out2);
\node at (-0.5,1) {$z'$};
\node at (6.5,2) {$=\ \  T^k_l(z')\, T^i_j(z) + \frac{ \hbar}{z-z'} T^i_l(z') \,T^k_j(z) \;.$};
\end{tikzpicture}
\end{center}

\begin{center}
\begin{tikzpicture}
\node(In1) at (0,4) {$\langle i|$};
\node(Out1) at (1,0) {$|j \rangle$};
\node(In2) at (1,4) {$\langle k|$};
\node(Out2) at (0,0) {$|l \rangle$};
\node(HLeft) at (-2,2) {};
\node(HRight) at (3,2) {};
\draw[->] (HRight) to (HLeft);
\draw[very thick, color=white] (In1) to (Out2);
\draw[very thick, color=white] (In2) to (Out1);
\draw[rounded corners] (In1.south) to (0,1.9) to  (Out1.north);
\node at (-0.5,3) {$z$};
\draw[color=white, very thick, rounded corners] (In2.south) to (1,1.9) to  (Out2.north);
\draw[rounded corners] (In2.south) to (1,1.9) to  (Out2.north);
\node at (1.2,3) {$z'$};
\node at (6.5,2) {$= \ \ T^i_j(z)\, T^k_l(z') +  \frac{\hbar}{z-z'} T^i_l(z)\, T^k_j(z) \:.$};
\end{tikzpicture}
\end{center}

From this we conclude that 
\begin{equation}
[T^k_l(z'), T^i_j(z) ]  =  \frac{\hbar}{z-z'} \left( - T^i_l(z')T^k_j(z) + T^i_l(z) T^k_j(z')  \right) \;.\label{RTT} 
\end{equation}
This equation is consistent at $z=z'$ because the expression in the brackets on the right hand  side has a zero when $z = z'$.  

We can rewrite \eqn \eqref{RTT} as an expression involving the generators $t^i_j[n]$.  We find that 
\begin{align}
\begin{split}
\sum_{m,n\ge 0}& (-1)^{n+m} (z')^{-n-1}z^{-m-1} \left[t^k_l[n], t^i_j[m]\right] \\
&= \frac{1}{z-z'}\sum_{n \ge 0} (z^{-n-1}  - (z')^{-n-1}) \left((-1)^n t^i_l[n]  \delta^k_j  - (-1)^n t^k_j[n]\delta^i_l  \right) \\
&+ \hbar \frac{1}{z - z'} \sum_{n,m \ge 0}(-1)^{n+m} \left( z^{-n-1} (z')^{-m-1} -   (z')^{-n-1}z^{-m-1}  \right)  \left( t^i_l[n] t^k_j[m]  \right) \;.
\end{split}
\end{align}
We can change $z \mapsto -z$, $z' \mapsto -z'$ to absorb the signs. 
The identity then becomes
\begin{multline}
 \sum_{m,n\ge 0} (z')^{-n-1}z^{-m-1} \left[t^k_l[n], t^i_j[m]\right] =\frac{1}{z-z'}\sum_{n \ge 0} (z^{-n-1}  - (z')^{-n-1}) \left( t^i_l[n]  \delta^k_j  -  t^k_j[n]\delta^i_l  \right) \\ 
-  \hbar \frac{1}{z - z'} \sum_{n,m \ge 0}\left(   z^{-n-1} (z')^{-m-1} -   (z')^{-n-1}z^{-m-1}  \right)   \left( t^i_l[n] t^k_j[m]  \right) \;.
\end{multline}

Using the identity
\begin{equation*}
\frac{z^{-n-1} -(z')^{-n-1} } {z-z'} = - (zz')^{-1} (z^{-n}  +  (z')^{-1} z^{-n+1} +  \dots +  (z')^{-n} ) \;,
\end{equation*}
we find that
the operators $t^i_j[n]$ satisfy the relation
\begin{align}
\begin{split}
&\sum_{n,m} (z')^{-n-1} z^{-m-1} \left[t^k_l[n], t^i_j[m]\right] =
\\
& \qquad - \sum_{n \ge 0} \left(z^{-n-1} (z')^{-1}  +z^{-n}  (z')^{-2} + \dots +  z^{-1} (z')^{-n-1} \right)  \left(t^i_l[n]  \delta^k_j  - t^k_j[n] \delta^i_l  \right) \\
& \qquad +  \hbar \sum_{m < n}\sum_{r=m+2}^{n+1}(z')^{-r} z^{r -(n+m+3)} \left(t^i_l[n] t^k_j[m] - t^i_l[m] t^k_j[n] \right)\;.
\end{split}
\end{align}
Equating coefficients gives us the commutation relations
\begin{align}
\begin{split}
&\left[t^k_l[n] , t^i_j[m] \right]  =  t^k_j[n+m] \delta^i_l - t^i_l[n+m]\delta^k_j \\
&\qquad -  \hbar \sum_{r+1 \le m,n } \left(t^i_l[r] t^k_j[m+n-1-r] - t^i_l[m+n-1-r] t^k_j[r]  \right)\;. \label{RTT_expanded} 
\end{split}
\end{align}
These are the commutations relations of $\mf{gl}_N[[z]]$, modified by a correction of order $\hbar$ that we will further discuss below. 

\subsection{Coproduct}
Next, let us describe the physical interpretation of the coproduct on the Yangian, in the RTT presentation.  The coproduct on the Yangian tells us how the Yangian algebra acts on the tensor product of two modules for the algebra. In terms of line operators, the coproduct on the Yangian will tell us how the generators $t^i_j[n]$ act on a horizontal Wilson line which is obtained by fusing two parallel horizontal Wilson lines.

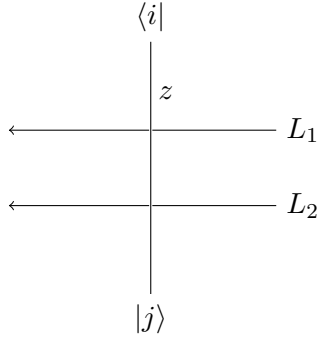
\begin{figure}[htbp]
\begin{center}
\begin{tikzpicture}
\node(In) at (0,4) {$\langle i|$};
\node(Out) at (0,0) {$|j \rangle$};
\node(HLeft) at (-2,2.5) {};
\node(HRight) at (2,2.5){$L_1$}; 
\node(HLeft2) at (-2,1.5) {};
\node(HRight2) at (2,1.5) {$L_2$};
\draw [<-] (HLeft) to (HRight);
 \draw [<-](HLeft2) to (HRight2); 
\draw[very thick, color=white] (In) to (Out);
\draw (In) to (Out);  
\node at (0.2,3) {$z$};
\end{tikzpicture}
\caption{\small{The configuration of Wilson lines associated to the coproduct.}}
\label{figure_coproduct}
\end{center}
\end{figure}

To understand this, let us consider a configuration (Fig. \ref{figure_coproduct}) of two horizontal line operators $L_1$ and $L_2$ in arbitrary representations
and one vertical line operator
in the fundamental representation, with chosen incoming and outgoing states on the vertical line.
We let $T^i_j(z, L_1 \otimes L_2)$ denote the operator that acts in this situation on the composite line operator obtained from fusing the two horizontal lines, and we write
 $T^i_j(z,L_1)$ and $T^i_j(z,L_2)$ for the corresponding operators on the individual Wilson lines. We let $\mc{H}_{L_i}$, $i=1,2$  and $\mc{H}_{L_1 \otimes L_2}$ denote the Hilbert spaces at the end of the individual horizontal Wilson lines and at the end of the fused Wilson line.  We have 
\begin{equation}
\mc{H}_{L_1 \otimes L_2} = \mc{H}_{L_1} \otimes \mc{H}_{L_2} \;.
\end{equation}
 The operators $T^i_j(z,L_i)$ and $T^i_j(z,L_1 \otimes L_2)$ are linear operators on the spaces $\mc{H}_{L_i}$, $\mc{H}_{L_1  \otimes L_2}$ respectively.

In \fig \ref{figure_coproduct},  we can move  the two horizontal Wilson lines in the vertical direction without changing anything, as long as the horizontal lines do not cross. When the lines are close together, the operator described by \fig \ref{figure_coproduct} is  $T^i_j(z,L_1 \otimes L_2)$.  When the lines are far apart, we can decompose the operator by summing over intermediate states placed on the vertical segment between the two horizontal lines. This yields $\sum_k T^i_k(z,L_1)\otimes T^k_j(z,L_2)$, where the sum over $k$ comes
from the fact that the segment of the vertical Wilson line between the two horizontal ones may carry an arbitrary label $k$.  We conclude that
\begin{equation}
T^i_j(z,L_1 \otimes L_2) = \sum_k T^i_k(z,L_1) \otimes T^k_j(z,L_2)\;.
\end{equation}
In the language of algebra, this identity tells us that the coproduct on the Yangian algebra (defined by the RTT  relation \eqref{RTT_glN}) is
\begin{equation}
\tr T^i_j(z) = \sum_k T^i_k(z) \otimes T^k_j(z) \;.
\end{equation}  
In terms of the expansion
\begin{equation}
T^i_j(z)  = \delta^i_j + \hbar \sum_{n \ge 0}z^{-n-1} (-1)^n  t^i_j[n] \;,
\end{equation}
the coproduct is
\begin{equation}
\label{equation_RTT_coproduct}
\tr t^i_j[n] = t^i_j[n]  \otimes 1 + 1 \otimes t^i_j[n] -  \hbar \sum_{r+s=n-1} t^i_k[r] t^k_j[s] \;.
\end{equation}
For $n=0,1$, this matches \cite{Part1}, \eqns (5.19) and (5.18) (with the structure constant given in \eqn \eqref{gl_commutation}), where we calculated the coproduct in the lowest
non-trivial order by a direct Feynman diagram calculation.  One surprising difference, however, is that the expression in \eqn \eqref{equation_RTT_coproduct} is \emph{exact}, and holds to all orders in $\hbar$.   This was not true for our calculation in \cite{Part1}, where we only performed a lowest order approximation to a Feynman diagram expansion that in principle could extend to all orders. 

\subsection{A Different Set of Generators}

There is an apparent discrepancy between the RTT presentation of the Yangian and the considerations of \cite{Part1}, section 5.  There we found that the
commutation relations of $\g[[z]]$ do not need any correction of order $\hbar$ in order to be consistent with the quantum-deformed coproduct.
Only a correction of $\hbar^2$  was needed.  In \eqn \eqref{RTT_expanded}, however, we found an $\O(\hbar)$ correction to the classical commutation relations in $\mf{gl}_N[[z]]$, with no $\O(\hbar^2)$ correction.

In this section, to reconcile the two approaches, we will see how to form a new set of generators built from linear plus quadratic expressions in the generators $t^i_j[n]$ such
that the algebra receives a correction only in order $\hbar^2$. In these new generators, the commutation relations in the algebra will not have any  quantum correction of
order $\hbar$; the correction of order $\hbar^2$  will match that studied in
\cite{Part1}, sections 5.4  and 8. The coproduct in these new generators will agree with \eqn \eqref{equation_RTT_coproduct} modulo $\hbar^2$, but will have higher-order corrections.  

The new generators are defined by the transformation
\begin{equation}
\til{t}^k_l[n] = t^k_l[n] + \hbar \sum_{r+s=n-1} t^\alpha_l[r] t^k_\alpha[s] \;.
\end{equation}
We will find that modulo $\hbar^2$, the commutator between two of the generators $\til{t}^k_l[n]$ is given by the classical commutation relations in the algebra $\mf{gl}_n[[z]]$. 

To show this, we first calculate that 
\begin{align}
\begin{split}
[\til{t}^k_l[n], \til{t}^i_j[m]] = &\delta^i_l t^k_j[n+m] - \delta^k_j t^i_l[n+m] \\
& -  \hbar   \sum_{r+1 \le m,n } \left(t^i_l[r] t^k_j[m+n-1-r] - t^i_l[m+n-1-r] t^k_j[r]  \right) \\
&+ \hbar  \sum_{r+s = n-1}\left(\delta^i_l  t^\alpha_j[m+r]t^k_\alpha[s]  -  t^i_l[m+r] t^k_j[s]\right) \\
&+ \hbar  \sum_{r+s=n-1} \left( t^i_l[r] t^k_j[s+m]- \delta^k_j t^\alpha_l[r]  t^i_\alpha[s+m] \right) \\ 
&+ \hbar  \sum_{r+s = m - 1}\left(  \delta^i_l t^\alpha_j[r] t^k_\alpha[n+s] - t^k_j[r]t^i_l[n+s] \right) \\
&+ \hbar   \sum_{r+s=m-1} \left(  t^k_j[n+r] t^i_l[s] - \delta^k_j t^\alpha_l[n+r]  t^i_\alpha[s]\right)  \\
&+ \O(\hbar^2)   \;.
\end{split}
\end{align}
In this expression we sum over repeated Greek indices.  The first two lines on the right are
the original commutation relations of the generators $t^i_j[n]$; the remaining lines describe the corrections coming from the difference between $\til{t}^i_j[n]$ and $t^i_j[n]$.

We can rearrange this sum to  give 
\begin{align}
\begin{split}
\left[ \til{t}^k_l[n], \til{t}^i_j[m] \right] &= \delta^i_l \til{t}^k_j[n+m] - \delta^k_j \til{t}^i_l[n+m] \\
& -  \hbar   \sum_{r+1 \le m,n } \left(t^i_l[r] t^k_j[m+n-1-r] - t^i_l[m+n-1-r] t^k_j[r]  \right) \\
&- \hbar  \sum_{r+s = n-1}\left(  t^i_l[m+r] t^k_j[s]-  t^i_l[r] t^k_j[s+m] \right) \\ 
&- \hbar  \sum_{r+s = m - 1}\left( t^k_j[r]t^i_l[n+s] -  t^k_j[n+r] t^i_l[s] \right) \\
& +  \O(\hbar^2) \;.
\label{equation_RTT_change} 
\end{split}
\end{align}
Noting that
\begin{equation} 
  t^k_j[r]t^i_l[n+s] -  t^k_j[n+r] t^i_l[s] =  t^i_l[n+s] t^k_j[r]-  t^i_l[s] t^k_j[n+r] + \O(\hbar)  \;,
 \end{equation}
we find that all the order $\hbar$ terms in \eqn (\ref{equation_RTT_change}) cancel.   Therefore the generators $\til{t}^k_l[n]$ only have an order $\hbar^2$ correction to the classical commutation relation. 

The first true quantum correction (that cannot be removed by a change of generators) occurs at order $\hbar^2$.    We will write explicitly the order $\hbar^2$ contribution to the commutation relation between elements of type $\til{t}^i_j[1]$.  We will focus on the commutator $[\til{t}^i_j[1], \til{t}^k_l[1]]$ under the assumption that the indices satisfy $\delta^i_l = 0$, $\delta^k_j = 0$.  This does not restrict the domain of validity of the result, because any terms in the commutator that involve $\delta^i_l$ or $\delta^k_l$ can be absorbed into a redefinition of the level $2$ generators.

One can calculate that
\begin{equation} 
	\left[\til{t}^k_l[1], \til{t}^i_j[1] \right]=  \delta^i_l \, \til{t}^k_j[2] - \delta^k_j \, \til{t}^i_l[2] +\hbar^2 \sum_{\alpha,\beta} \left[ \tilde t^\alpha_l [0]\, \tilde t^k_\alpha[0],\tilde  t^\beta_j [0]  \, \tilde t^i_\beta[0]  \right] +\O(\hbar^3)  \;.
\end{equation}
   Thus  the $\tilde t$'s satisfy the commutation relations of $\gl_N[[z]]$ modulo a
correction of order $\hbar^2$. Further, the coefficient of $\hbar^2$ on the right hand side matches the order $\hbar^2$ contribution to the commutation relations computed in \cite[section 8]{Part1}, up to a term that can be re-absorbed into the definition of $\til{t}^r_s[2]$. This is expected, since as explained in \cite[section 5.4]{Part1} the coefficient of  the $\hbar^2$ term 
can be fixed by the Jacobi identity and the expression for the coproduct given in 
\eqn \eqref{equation_RTT_coproduct}. 

\section{\texorpdfstring{RTT Presentation for $Y(\mf{so}_N)$ and $Y(\mf{sp}_{2N})$}{RTT Relation for Y(so(N)) and Y(sp(2N))}}\label{RTT_so_sp}

We have seen how the Yangian for $\mf{gl}_N$ arises in a natural way from our set-up. It is natural to ask whether one can derive a similar expression for the Yangian for other groups. In this section we will do this for the Lie algebras $\mf{so}_N$ and $\mf{sp}_{2N}$.   

Let us start by defining some  RTT relations for a general simple Lie algebra
$\mathfrak{g}$ and an arbitrary irreducible representation $V$, along the lines of the definition given in section \ref{RTT_gl} for the fundamental representation of $\mf{gl}_N$.   We assume that the Wilson line defined classically using $V$ lifts to a quantum Wilson line.  
   
Choose a basis $e_i$ of $V$.  The RTT algebra is generated by the coefficients of the series \eqref{T_series}.
The generators $t^i_j[n]$ are subject to the relation 
\begin{equation}\label{RTTrel}
\sum_{r,s} R^{i k}_{r s}(z - z') T^r_j(z') T^s_l(z) = \sum_{r,s} T^i_r(z) T^k_s(z')R^{rs}_{jl}(z -z') \;.
\end{equation}
where $R^{ik}_{rs}(z-z')$ are the matrix entries for the R-matrix associated to two copies of the representation $V$, one placed at $z$ and one at $z'$. 

As in the analysis in section \ref{RTT_gl}, the RTT algebra associated to the representation $V$ acts on the space of states at the end of any horizontal Wilson line. The action is defined by considering a configuration of Wilson lines like 

\begin{center}
\begin{tikzpicture}
\node(In) at (0,4) {$\langle i|$};
\node(Out) at (0,0) {$|j \rangle$};
\node(HLeft) at (-2,2) {};
\node(HRight) at (2,2) {};
\draw (HLeft) to (HRight); 
\draw[very thick, color=white] (In) to (Out);
\draw (In) to (Out);  
\node at (0.2,3) {$z$};
\end{tikzpicture}
\end{center}
where the vertical Wilson line is associated to the representation $V$ and has incoming and outgoing states $i$ and $j$. The horizontal Wilson line is associated to the representation $W$. 

When we studied the case of 
$\mathfrak{g}=\gl_N$, we saw that the operators $t^i_j[k] : W \to W$ were quantizations of the operators whereby the derivative $\tfrac{1}{k!}\partial_z^k A$ of the gauge field is coupled to the horizontal Wilson line.  The RTT relation then told us the conditions that these operators must satisfy at the quantum level in order to have a consistent Wilson line. 

We would like to be able to make this statement for a general group, but there is an immediate problem.  There are many more operators $t^i_j[k]$ then there are generators of the Lie algebra $\mf{g}[[z]]$ whose representations describe classical Wilson lines.  For each $k$,
 there are $D^2$ of these operators, where $D$ is the dimension of the representation $V$,
 while the number we want is the dimension of 
the Lie algebra $\g$. Regardless of the choice of $V$, $D^2$ will be too large (except for $\mf{gl}_N$), so more relations
 are needed to remove the extra generators.  At the classical level, the extra relations are familiar.  For example, suppose
 that $V$ is the fundamental $N$-dimensional representation of $\mf{so}_N$ or $\mf{sl}_N$.  The RTT relation alone
 will give $N^2$ generators $t^i_j[k]$, for each $k$.  At the classical level, $\mf{so}_N$ is generated by antisymmetric $N\times N$
 matrices and $\mf{sl}_N$ by traceless ones, so the extra relations on $t^i_j[k]$ are $t^i_j[k]+t^j_i[k]=0$ for $\mf{so}_N$,
 or $\sum_i t^i_i[k]=0$ for $\mf{sl}_N$. 

The goal of this section is to see, from the point of view of field theory, how to introduce extra relations which will remove the redundant generators in the case of the Yangian algebras of
$\so_N$ and $\sp_{2N}$.   Other Lie algebras are considered in sections \ref{RTT_sl} and \ref{RTTexc}. 
  As a matter of terminology, we will refer to eqn. (\ref{RTTrel}) as the
RTT relation, while a set of relations including this one that gives a complete description of an algebra will be called an RTT presentation of that algebra.

Let $V$ be the fundamental representation of $\mf{so}_N$ or $\mf{sp}_{2N}$, equipped with its symmetric or antisymmetric invariant pairing. 
We can consider a curved Wilson line in the $xy$ plane labeled by this representation, but when we do that, we have to take into account the framing anomaly.
In \fig \ref{figure_curved}, we depict a curved Wilson line in the representation $V$, which is curved so that asymptotically both of its ends run vertically downwards.  The
endpoints of this line are labeled by states that are basis vectors of $V$.  Note that, because the representation $V$ is self-dual, each end is labeled by the same representation
$V$.   This line crosses in the horizontal direction a Wilson line labeled by an arbitrary representation $W$.  
The operator that acts on the representation $W$ is, as shown in the figure,
 the product $ T^k_i(z) T^l_j(z + \hbar \sh^\vee) \omega_{kl}$, where $\omega_{kl}$ indicates the pairing on $V$ 
 (which is either symmetric or antisymmetric depending on whether we work with $\mf{so}_N$ or $\mf{sp}_{2N}$). The shift in one factor from $z$ to $z+\hbar\sh^\vee$
 reflects the framing anomaly.

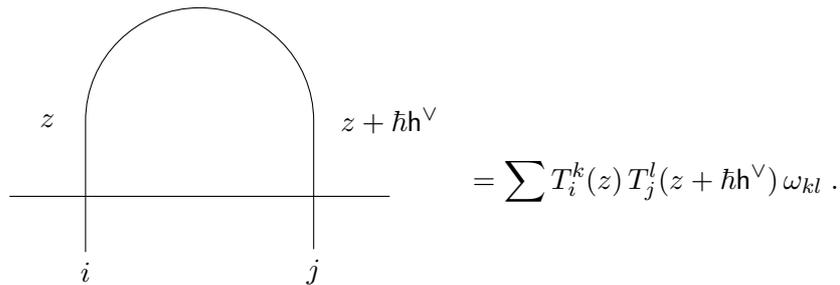
\begin{figure}[htbp]
\begin{center}
\begin{tikzpicture}
\node(N1) at (-1.5,-2) {$i$};
\node(N2) at (1.5,-2) {$j$};
\draw (N1) to (-1.5,0) to [out=90,in=180] (0,1.5) to [out=0,in = 90] (1.5,0) to (N2);
\node at (-2,0) {$z$}; 
\node at (2.5,0) {$z + \hbar \sh^\vee$};
\draw (-2.5,-1) to (2.5,-1);
\node at (6,-0.8) {$= \displaystyle\sum T^k_i(z) \, T^l_j(z + \hbar \sh^\vee) \, \omega_{kl} \;. $}; 
\end{tikzpicture}
\caption{\small{A curved Wilson line in representation $V$, with its ends labeled by basis vectors $i$ and $j$,  crossing a horizontal Wilson line in an arbitrary representation $W$.  The induced operator acting on
$W$ can be written as shown in terms of a product of $R$-matrices.  In doing so, one has to take into account the framing anomaly.   }}
\label{figure_curved}
\end{center}
\end{figure}

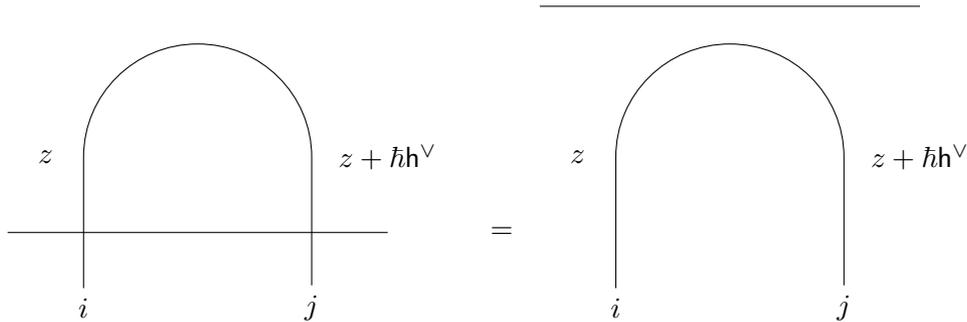
\begin{figure}[htbp]
\begin{center}
\begin{tikzpicture}
\node(N1) at (-1.5,-2) {$i$};
\node(N2) at (1.5,-2) {$j$};
\draw (N1) to (-1.5,0) to [out=90,in=180] (0,1.5) to [out=0,in = 90] (1.5,0) to (N2);
\node at (-2,0) {$z$}; 
\node at (2.5,0) {$z + \hbar \sh^\vee$};
\draw (-2.5,-1) to (2.5,-1);
\node at (4,-1) {$=$};
\begin{scope}[shift={(7,0)}]
\node(N1) at (-1.5,-2) {$i$};
\node(N2) at (1.5,-2) {$j$};
\draw (N1) to (-1.5,0) to [out=90,in=180] (0,1.5) to [out=0,in = 90] (1.5,0) to (N2);
\node at (-2,0) {$z$}; 
\node at (2.5,0) {$z + \hbar \sh^\vee$};
\draw (-2.5,2) to (2.5,2);
\end{scope}
\end{tikzpicture}
\caption{\small{The straight horizontal Wilson line can be moved in the vertical direction to avoid intersecting the curved  and asymptotically vertical one. }}
\label{moving}
\end{center}
\end{figure}

Now as in \fig \ref{moving}, we can move the horizontal Wilson line vertically so that it no longer intersects the curved one.  The equivalence of the two
configurations gives us the identity 
\begin{equation}
\sum T^k_i\left(z\right) T^l_j\left(z + \hbar \sh^\vee\right)  \omega_{kl}  = \omega_{ij} \label{equation_rtt_sosp}  \;.
\end{equation}

If we specialize to $\mf{so}_N$ with its fundamental representation, we find  
\begin{equation}
\sum_{k}  T^{k}_{i}\left(z \right) T^{k}_{j}\left(z + (N-2) \hbar \right)= \delta_{i j}\;, \label{relation_so} 
\end{equation}
where we use the fact that the dual Coxeter number $\sh^\vee$ of $\mf{so}_N$ is $N-2$ for $N>4$.  

If we write this expression in terms of the generators $t^i_j[r]$, we find that the coefficient of $\hbar z^{-r-1}$ gives us the relation
\begin{equation}\label{ridi}
t^{i_1}_{i_0}[r] +  t^{i_0}_{i_1}[r] = 0  \quad(\text{modulo } \hbar) \;,
\end{equation} 
so that the operators $t^i_j[r]$ are skew-symmetric.  This tells us that, modulo $\hbar$, the algebra obtained from supplementing the RTT relation with eqn.\ \eqref{relation_so} is the universal enveloping algebra of $\mf{so}_N[[z]]$.   Concretely, the RTT relation alone would give us an algebra of the size of $\gl_N[[z]]$  (but with a different algebra structure\footnote{The different algebra structure appears even semi-classically because the $R$-matrix involves the Casimir for $\mf{so}_N$ instead of $\gl_N$.}).  The new relation (\ref{ridi}) reduces this to $\so_N[[z]]$.

We can similarly consider $\mf{sp}_{2N}$. Let $V$ denote the $2N$ dimensional vector representation, and choose a Darboux basis $e_i$ of $V$ with pairing $\omega_{ij}$. The dual Coxeter number of $\mf{sp}_{2N}$ is $N+1$.  Eqn. \eqref{equation_rtt_sosp} takes the form
\begin{equation}
\sum T^k_i\left(z\right) T^l_j\left(z + \hbar (N+1)\right)  \omega_{kl}  = \omega_{ij} \;.
\end{equation} 
This relation implies that classically the generators $t^i_j[r]$ form a copy of the algebra $\mf{sp}_{2N}[[z]]$. 

These relations are known in the literature \cite{Drinfeld_Hopf,Drinfeld_ICM} to give presentations of the Yangians for $\mf{so}_N$ and of $\mf{sp}_{2N}$.  


\section{\texorpdfstring{RTT Presentation for $Y(\sl_N)$}{RTT Relation for Y(sl(N))}}
\label{RTT_sl}

We have seen in section \ref{RTT_gl} how to describe the Yangian algebra associated to the group $\gl_N$, using the RTT relation.  We have also derived in section \ref{RTT_so_sp} extra relations that we can add to the RTT relation to get a presentation of the Yangian algebra associated to $\so_N$ and $\sp_{2N}$. In this section we will use the results from \cite{Part1}, section
7 on networks of Wilson lines to derive an extra relation that one can add to the Yangian of $\gl_N$ to get the Yangian for $\sl_N$.    

We will find in general that there is always an extra relation that we can add to the RTT relation whenever we have a vertex connecting multiple copies of the same Wilson line. To start with, we will describe this extra relation.

Let us consider the RTT algebra associated to a Lie algebra $\mathfrak{g}$ and a representation $V$ as defined in section \ref{RTT_so_sp}.  We assume that $V$ lifts to a Wilson line at the quantum level.  Suppose that  $v \in  V^{\otimes n}$ is a $\g$-invariant vector, and that we can quantize $v$ to a vertex linking $n$ copies of the Wilson line. For simplicity we will assume that $v$ is cyclically invariant or anti-invariant, so that the angles between successive Wilson lines are $2 \pi / n$. 

\begin{figure}[htbp]
\begin{center}
\label{figure_anomaly_free_network}
\begin{tikzpicture}
\node[draw, circle] (central) at (0:0) {$v$};
\node(N1) at (-5,-3) {$z$};
\node(N2) at (-3,-3) {$z + \frac{1}{6} 2 \hbar \sh^\vee$};
 \node(N3) at (-1,-3) {$z + \frac{2}{6} 2 \hbar \sh^\vee$};
\node(N4) at (1,-3){$z + \frac{3}{6} 2 \hbar \sh^\vee$};
\node(N5) at (3,-3){$z + \frac{4}{6} 2 \hbar \sh^\vee$};
\node(N6) at (5,-3){$z + \frac{5}{6} 2 \hbar \sh^\vee$};
\draw[-<-] (N1) to (-5,-2) to [out=90,in=120] (central);
\draw[-<-] (N2) to (-3,-2) to [out=90,in=180] (central);
\draw[-<-] (N3) to (-1,-2) to [out=90,in=240] (central);
\draw[-<-] (N4) to (1,-2) to [out=90,in=300] (central);
\draw[-<-] (N5) to (3,-2) to [out=90,in=0] (central);
\draw[-<-] (N6) to (5,-2) to [out=90,in=60] (central);
 \end{tikzpicture}
\caption{\small{This is a diagram of a consistent quantum configuration built from an invariant tensor $v \in  V^{\otimes n}$ which is cyclically symmetric or antisymmetric,
so that angles between successive lines at the vertex are $2\pi/n$.  
  The lines emanating from the vertex have been bent so as to all run asymptotically downwards.  The framing anomaly dictates the relative values of $z$ at the bottom of the figure, where all lines are vertical.  At the vertex all lines have the same value of $z$.  Sketched here is the case $n=6$. } }\label{vertex}
\end{center}
\end{figure}
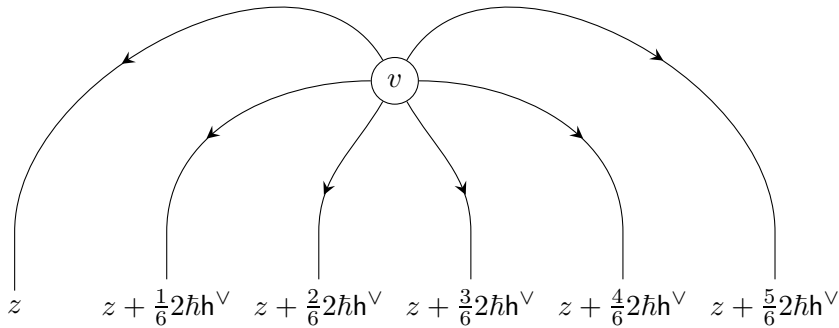 

Given such a vertex, we can bend the $n$ lines emanating from the vertex so that they all run asymptotically downwards in the $xy$ plane, as in Fig. \ref{vertex}.
As usual, when we do this, we have to take into account the framing anomaly.  Successive lines are displaced in $z$ by $(2/n)\hbar\sh^\vee$, as indicated in Fig. \ref{vertex}.
We can now consider  a horizontal Wilson line, in an arbitrary representation $W$, as in \fig \ref{figure_crossing_network}.   If the horizontal Wilson line is below the vertex
(and therefore intersects all $n$ lines that emanate from it), then the matrix acting on the representation $W$ is a product of $n$ $R$-matrices, as indicated in the figure.

\begin{figure}[htbp]
\begin{center}
\begin{tikzpicture}

\node[draw, circle] (central) at (0:0) {$v$};
\node at (-3,-2) {$z_0$};
\node at (-2,-2) {$z_1$};
\node at (-1,-2) {$z_2$};
\node at (0,-2) {$z_3$};
\node at (1,-2) {$z_4$};
\node at (2,-2) {$z_5$};

\node at (-2.5,-4) {$i_0$};
\node at (-1.5,-4) {$i_1$};
\node at (-0.5,-4) {$i_2$};
\node at (0.5,-4) {$i_3$};
\node at (1.5,-4) {$i_4$};
\node at (2.5,-4) {$i_5$};

\node at (100:0.7) {$k_0$};
\node at (160:0.7) {$k_1$};
\node at (220:0.7) {$k_2$};
\node at (280:0.7) {$k_3$};
\node at (340:0.7) {$k_4$};
\node at (35:0.7) {$k_5$};

\node(N1) at (-2.5,-3.5) {};
\node(N2) at (-1.5,-3.5){}; 
 \node(N3) at (-0.5,-3.5){}; 
\node(N4) at (0.5,-3.5){};
\node(N5) at (1.5,-3.5){};
\node(N6) at (2.5,-3.5){};
\draw[-<-] (N1) to (-2.5,-2) to [out=90,in=120] (central);
\draw[-<-] (N2) to (-1.5,-2) to [out=90,in=180] (central);
\draw[-<-] (N3) to (-0.5,-2) to [out=90,in=240] (central);
\draw[-<-] (N4) to (0.5,-2) to [out=90,in=300] (central);
\draw[-<-] (N5) to (1.5,-2) to [out=90,in=0] (central);
\draw[-<-] (N6) to (2.5,-2) to [out=90,in=60] (central);
\draw (-3.5, -2.5) to (3.5,-2.5);

\node at (4,-1.5) {$=$};
\node at (7.5,-1.5){$ \sum_{k_0,\dots,k_6} v_{k_0,\dots,k_5} T^{k_0}_{i_0}(z_0) \dots T^{k_5}_{i_5}(z_5) \;.$ };

\end{tikzpicture}
\caption{\small{This figure shows the effect of a horizontal Wilson line crossing a network of Wilson lines coming from an invariant tensor $v \in  V^{\otimes 6}$. All the vertical Wilson lines are in the representation $V$.  The framing anomaly dictates that $z_i = z_0 + \frac{i}{n} 2 \hbar \sh^\vee$.}}
\label{figure_crossing_network}
\end{center}
\end{figure}
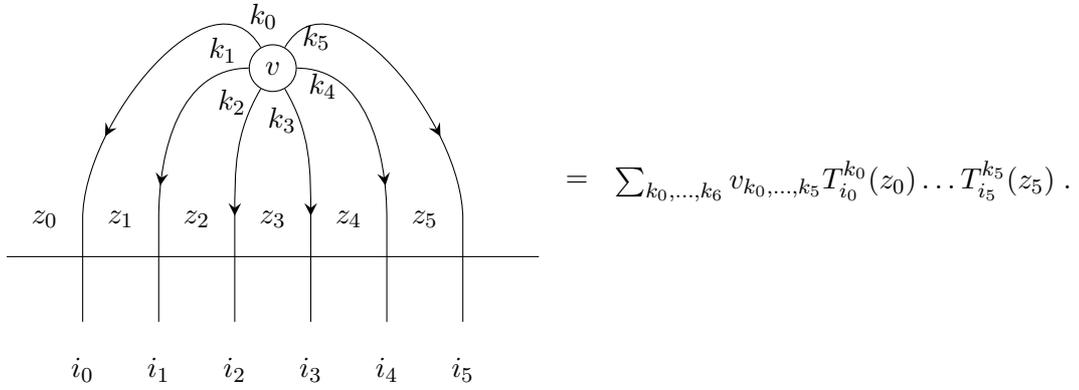
But we are free as in Fig. \ref{figure_identity} to move the horizontal Wilson line so as not to intersect the vertical network at all.    This freedom leads to the
following identity:  \begin{equation}
\sum_{k_r} v_{k_0,\dots,k_{n-1}} T^{k_0}_{i_0}\left(z \right) T^{k_1}_{i_1}\left(z + \tfrac{2}{n} \hbar \sh^\vee \right)\cdots T^{k_{n-1}}_{i_{n-1}}\left(z + \tfrac{2(n-1)}{n} \hbar \sh^\vee\right) = v_{i_0,\dots,i_{n-1}}\;. \label{relation_invariant_tensor}
\end{equation} 
Note that eqn. (\ref{equation_rtt_sosp}) can be viewed as the special case of this identity with $n=2$.

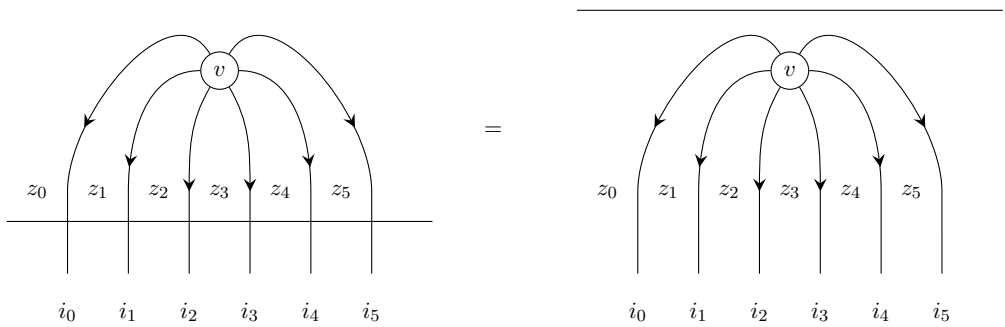
\begin{figure}
\begin{center}
\begin{tikzpicture}
\begin{scope}[scale=0.8, every node/.style={transform shape}]
\node[draw, circle] (central) at (0:0) {$v$};
\node at (-3,-2) {$z_0$};
\node at (-2,-2) {$z_1$};
\node at (-1,-2) {$z_2$};
\node at (0,-2) {$z_3$};
\node at (1,-2) {$z_4$};
\node at (2,-2) {$z_5$};
\node at (-2.5,-4) {$i_0$};
\node at (-1.5,-4) {$i_1$};
\node at (-0.5,-4) {$i_2$};
\node at (0.5,-4) {$i_3$};
\node at (1.5,-4) {$i_4$};
\node at (2.5,-4) {$i_5$};
\node(N1) at (-2.5,-3.5){};
\node(N2) at (-1.5,-3.5){}; 
 \node(N3) at (-0.5,-3.5){}; 
\node(N4) at (0.5,-3.5){};
\node(N5) at (1.5,-3.5){};
\node(N6) at (2.5,-3.5){};
\draw[-<-] (N1) to (-2.5,-2) to [out=90,in=120] (central);
\draw[-<-] (N2) to (-1.5,-2) to [out=90,in=180] (central);
\draw[-<-] (N3) to (-0.5,-2) to [out=90,in=240] (central);
\draw[-<-] (N4) to (0.5,-2) to [out=90,in=300] (central);
\draw[-<-] (N5) to (1.5,-2) to [out=90,in=0] (central);
\draw[-<-] (N6) to (2.5,-2) to [out=90,in=60] (central);
\draw (-3.5, -2.5) to (3.5,-2.5);
\node at (4.5,-1) {$=$};
\end{scope}

\begin{scope}[shift={(7.5,0)},scale=0.8,  every node/.style={transform shape}]
\node[draw, circle] (central) at (0:0) {$v$};
\node at (-3,-2) {$z_0$};
\node at (-2,-2) {$z_1$};
\node at (-1,-2) {$z_2$};
\node at (0,-2) {$z_3$};
\node at (1,-2) {$z_4$};
\node at (2,-2) {$z_5$};
\node at (-2.5,-4) {$i_0$};
\node at (-1.5,-4) {$i_1$};
\node at (-0.5,-4) {$i_2$};
\node at (0.5,-4) {$i_3$};
\node at (1.5,-4) {$i_4$};
\node at (2.5,-4) {$i_5$};
\node(N1) at (-2.5,-3.5){};
\node(N2) at (-1.5,-3.5){}; 
 \node(N3) at (-0.5,-3.5){}; 
\node(N4) at (0.5,-3.5){};
\node(N5) at (1.5,-3.5){};
\node(N6) at (2.5,-3.5){};
\draw[-<-] (N1) to (-2.5,-2) to [out=90,in=120] (central);
\draw[-<-] (N2) to (-1.5,-2) to [out=90,in=180] (central);
\draw[-<-] (N3) to (-0.5,-2) to [out=90,in=240] (central);
\draw[-<-] (N4) to (0.5,-2) to [out=90,in=300] (central);
\draw[-<-] (N5) to (1.5,-2) to [out=90,in=0] (central);
\draw[-<-] (N6) to (2.5,-2) to [out=90,in=60] (central);
\draw (-3.5, 1) to (3.5,1);
\end{scope}
\end{tikzpicture}
\end{center}
\caption{\small{Topological invariance allows us to move the position of the horizontal Wilson line without effecting the result\label{figure_identity} }}
\end{figure}

Let us specialize the identity \eqref{relation_invariant_tensor} to the case that the vector space $V$ which labels the vertical Wilson lines is the fundamental representation of $\sl_N$, and that $v \in V^{\otimes N}$ is a totally antisymmetric element. This configuration can be quantized,
as shown in \cite{Part1}. The dual Coxeter number $\sh^\vee$ of $\sl_N$ is $N$.  We therefore find the identity\footnote{What we are here calling $\op{Alt}(k_0,\dots,k_{N-1})$ is often written $\varepsilon_{k_0k_1\cdots k_{N-1}}$, where $\varepsilon$ is the antisymmetric
Levi-Civita tensor.} 
\begin{equation}
\sum_{k_r} \op{Alt}(k_0,\dots,k_{N-1}) T^{k_0}_{i_0}(z ) T^{k_1}_{i_1}(z + 2 \hbar )\cdots T^{k_{N-1}}_{i_{N-1}}(z + 2(N-1) \hbar ) = \op{Alt}(i_0,\dots,i_{N-1})\;,
\end{equation}  
where  $\op{Alt}$ is the alternating symbol.  

Specializing further to the case that $i_0 = 0, i_1 = 1,\dots, i_{N-1} = N-1$ (where our basis of $V = \C^N$ is $e_0,e_1,\dots, e_{N-1}$)  we find the identity 
\begin{equation}
\sum_{k_r} \op{Alt}(k_0,\dots,k_{N-1}) T^{k_0}_{0}(z ) T^{k_1}_{1}(z + 2 \hbar  )\cdots T^{k_{N-1}}_{N-1}(z + 2(N-1) \hbar ) = 1\;. \label{equation_qdet}  
\end{equation} 
In the literature \cite{Chari-Pressley}, this relation is often written as
\begin{equation}
\op{qDet}(T(z)) = 1\;,
\end{equation}
where $\op{qDet}$ is called the quantum determinant.  
It is known that this relation, together with the RTT relation, is a complete presentation of the Yangian
 for $\mf{sl}_N$ \cite{Drinfeld_original, Chari-Pressley}. 

Let us verify that this relation implies that, modulo $\hbar$, $\sum t^i_i[k] = 0$ for each $k$.  To check this, we use the expansion in \eqn \eqref{T_series}.
Then the coefficient of $z^{-n-1} \hbar$ in \eqn \eqref{equation_qdet} tells us that
\begin{equation}
\sum t^i_i[n] = 0  \quad (\text{modulo } \hbar) \;.
\end{equation}
Thus, modulo $\hbar$, the operators $t^i_j[n]$ provide a basis for the Lie algebra $\mf{sl}_N[[z]]$.  The higher-order terms in eqn. \eqref{equation_qdet}, together with the RTT
relation, provide quantum corrections.

\section{\texorpdfstring{RTT Presentation for $Y(\mathfrak{g}_2), Y(\mathfrak{f}_4), Y(\mathfrak{e}_6)$ and $Y(\mathfrak{e}_7)$}{RTT Presentation for Y(g(2)), Y(f(4)), Y(e(6)) and Y(e(7))}}\label{RTTexc}


In our analysis of the RTT presentation of the Yangian for $\mf{sl}_N$, we saw how the vertex connecting $n$ copies of the fundamental representation of $\mf{sl}_N$ led to an extra relation in the RTT algebra: the quantum determinant.  In a similar way, we will use the cubic invariant tensor in the fundamental representation of $\gtwo$, $\ffour$ and $\esix$ to give an RTT presentation of these algebras. For $\eseven$ we will use the invariant tensor in the fourth power of the fundamental representation.

Unfortunately, we were unable to find an RTT presentation of $\eeight$, because this algebra does not have a representation with convenient properties.
The lowest-dimensional non-trivial representation of $\eeight$ is the adjoint representation
$\mathbf{248}$.  The corresponding Wilson line does not quantize: the two-loop anomaly studied in \cite{Part1}, section 8, does not vanish for this representation.  
The direct sum $\mathbf{248}\oplus \mathbf{1}$ of the adjoint with a trivial representation does quantize, but the methods of \cite{Part1} do not lead to an easy construction of
vertices for this representation, as one would wish in order to construct an RTT presentation.   Considerations of integrable $S$-matrices suggest that suitable vertices exist, but we will
not explore this direction.

Instead, let us start with $\gtwo$, $\ffour$  and $\esix$, where the methods of this paper do apply conveniently.  Let $V$ be the lowest-dimensional representation of one of these groups, so that $V$ is either the $\mathbf{7}$ of $\gtwo$, the $\mathbf{26}$ of $\ffour$, or the $\mathbf{27}$ of $\esix$. 
In  each case there is an essentially unique invariant  $v\in V^{\otimes 3}$, which moreover  is cyclically invariant or anti-invariant.
 We have already observed in \cite{Part1}, section 7.6, that in each case these representations quantize to line operators in our theory, and that likewise the invariant $v$
 quantizes to a trivalent vertex with relative angles $2\pi/3$.
 
 We will consider the three cases in turn, and then move on to $\eseven$.

\subsection{\texorpdfstring{$Y(\gtwo)$}{Y(g(2))}}
The Lie algebra $\mf{g}_2$ is the endomorphisms of its fundamental representation $\mathbf{7}$ which preserve both the symmetric invariant pairing and the cubic antisymmetric invariant tensor.  We can quantize this description of $\mf{g}_2$ to give an RTT description of the Yangian for $\mf{g}_2$.  

As in our previous examples, the
 RTT relations come from an analysis of a vertical Wilson line in the $\mbf{7}$ crossing an arbitrary horizontal Wilson line.  Thus, the RTT algebra is generated by the coefficients of a series \eqref{T_series},
where the indices $i,j$ now run from $1$ to $7$ and indicate a basis of the $\mathbf{7}$ of $\gtwo$ which is orthonormal with respect to the $\gtwo$-invariant inner product. 

These generators are subject to the RTT relation
\begin{equation}
\sum_{r,s} R^{i k}_{r s}(z - z') T^r_j(z') T^s_l(z) = \sum_{r,s} T^i_r(z) T^k_s(z')R^{rs}_{jl}(z -z') \;,
\end{equation}
where $R$ is the $R$-matrix for the crossing of two Wilson lines in the $\mathbf{7}$.

We add two extra relations to this RTT relation. One comes from the fact that the representation $\mathbf{7}$ is self-dual. By an analysis identical to that in section \ref{RTT_so_sp}, this
 self-duality gives rise to the relation
\begin{equation}\label{igloo}
\sum_{k}  T^{k}_{i}(z ) T^{k}_{j}(z +  \hbar \sh^\vee )= \delta_{i j} \;,
\end{equation}
where $\sh^\vee = 4$ is the dual Coxeter number of $\gtwo$. 

As in our analysis of $\so_N$ and $\sp_{2N}$, if we write we write this relation (\ref{igloo})  in terms of the generators $t^i_j[r]$, we find that the coefficient of $\hbar z^{-r-1}$ gives us the relation
\begin{equation}
t^{i_1}_{i_0}[r] +  t^{i_0}_{i_1}[r] = 0 \quad (\text{modulo } \hbar) \;,
\end{equation} 
so that the operators $t^i_j[r]$ are skew-symmetric. 

Next, let us add an additional relation coming from the invariant tensor in $\mathbf{7}^{\otimes 3}$.  This relation comes as in our discussion of the quantum
determinant of $\sl_N$ from the following picture:

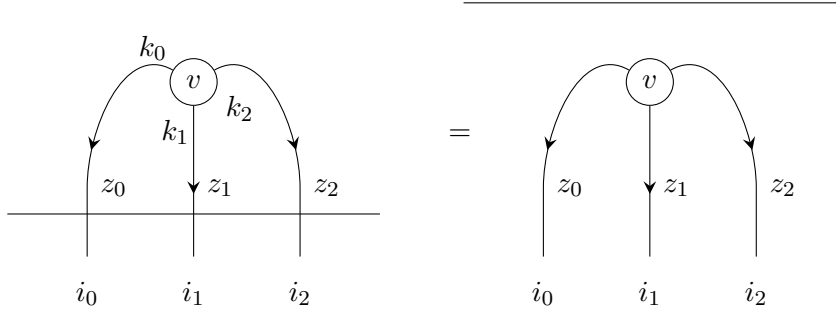
\begin{center}
\begin{tikzpicture}
\begin{scope}[shift={(-4,0)},scale={(0.7)}]
\node[draw, circle] (central) at (0:0) {$v$};
\node at (-1.5,-2) {$z_0$};
\node at (0.5,-2) {$z_1$};
\node at (2.5,-2) {$z_2$};

\node at (-2,-4) {$i_0$};
\node at (0,-4) {$i_1$};
\node at (2,-4) {$i_2$};
\node(N1) at (-2,-3.5) {};
\node(N2) at (0,-3.5){}; 
 \node(N3) at (2,-3.5){}; 
\node at (140:1) {$k_0$};
\node at (250:1) {$k_1$};
\node at (330:1) {$k_2$};
\draw[-<-] (N1) to (-2,-2) to [out=90,in=150] (central);
\draw[-<-] (N2) to (0,-2) to [out=90,in=270] (central);
\draw[-<-] (N3) to (2,-2) to [out=90,in=30] (central);
\draw (-3.5, -2.5) to (3.5,-2.5);
\node at (5,-1) {$=$};
 \end{scope}
\begin{scope}[shift={(2,0)},scale={(0.7)}]
\node[draw, circle] (central) at (0:0) {$v$};
\node at (-1.5,-2) {$z_0$};
\node at (0.5,-2) {$z_1$};
\node at (2.5,-2) {$z_2$};
\node at (-2,-4) {$i_0$};
\node at (0,-4) {$i_1$};
\node at (2,-4) {$i_2$};
\node(N1) at (-2,-3.5) {};
\node(N2) at (0,-3.5){}; 
\node(N3) at (2,-3.5){}; 
\draw[-<-] (N1) to (-2,-2) to [out=90,in=150] (central);
\draw[-<-] (N2) to (0,-2) to [out=90,in=270] (central);
\draw[-<-] (N3) to (2,-2) to [out=90,in=30] (central);
\draw (-3.5, 1.5) to (3.5,1.5);
\end{scope}
\end{tikzpicture}
\end{center}
The framing anomaly, together with the fact that the Wilson lines form angles of $2 \pi / 3$ at the vertex, tells us that the $z$-values of the Wilson lines (when they are vertical) must satisfy
\begin{align}
\begin{split}
z_1 &= z_0 + \tfrac{2}{3} \hbar \sh^\vee 
   = z_0 + \tfrac{8}{3} \hbar \;, \\
z_2 &= z_0 + \tfrac{4}{3} \hbar \sh^\vee 
= z_0 + \tfrac{16}{3} \hbar \;.
\end{split}
\end{align}
If $\Omega_{ijk}$ denotes the invariant tensor in $\mbf{7}^{\otimes 3}$, then this picture leads to the relation
\begin{equation}
\sum \Omega_{k_0 k_1 k_2} T_{i_0}^{k_0}(z) \, T_{i_1}^{k_1}\left(z + \tfrac{2}{3} \hbar \sh^\vee\right) T_{i_2}^{k_2}\left(z + \tfrac{4}{3} \hbar \sh^\vee\right) = \Omega_{i_0 i_1 i_2} \;.
\end{equation}
If we expand this relation out in terms of the operators $t^i_j[r]$, we find that
\begin{equation}
\Omega_{k i_1 i_2} t_{i_0}^k[r] + \Omega_{i_0 k i_2}t_{i_1}^{k}[r] + \Omega_{i_0 i_1 k} t_{i_2}^k [r] = 0 \quad (\text{modulo } \hbar) \;.
\end{equation}
This is the condition that the elements $t^i_j[r]$ preserve the tensor $\Omega$ modulo $\hbar$. Together with the fact that they are antisymmetric, it follows that modulo $\hbar$ the elements $t^i_j[r]$ live in $\mf{g}_2[[z]]$.  

As in section \ref{RTT_gl}, the RTT relation by itself would imply that the generators $t^i_j[r]$ satisfy  an algebra
with the size of $\gl_7[[z]]$.  The other relations simply restrict mod $\hbar$
to $\g_2[[z]]$.

\subsection{\texorpdfstring{$Y(\ffour)$ and $Y(\esix)$}{Y(f(4)) and Y(e(6))}}
The $\mathbf{26}$ of $\ffour$ is real, so it has a symmetric invariant bilinear pairing, and the Lie algebra $\mf{f}_4$ is the algebra of endomorphisms of the $\mathbf{26}$ which preserve this symmetric pairing and also the invariant $3$-tensor in $\mathbf{26}^{\otimes 3}$.  As in our discussion of the Yangian of $\mf{g}_2$, we find an RTT presentation of the Yangian for $\mf{f}_4$ by imposing the quantum version of these relations.  Let us choose an orthonormal basis of the representation $\mathbf{26}$, and let $\Omega_{ijk}$ denote the invariant cubic tensor in this  basis.  Then, the Yangian for $\mf{f}_4$ is generated by the coefficients of a series \eqref{T_series} 
where $i,j$ now run from $1$ to $26$.  The relations are
\begin{align}
\sum_{r,s} R^{i k}_{r s}(z - z') T^r_j(z') T^s_l(z) &= \sum_{r,s} T^i_r(z) T^k_s(z')R^{rs}_{jl}(z -z') \;, \\ 
\sum_{k}  T^{k}_{i}(z ) T^{k}_{j}(z +  \hbar \sh^\vee )&= \delta_{i j} \;, \\
\sum_{k_0,k_1,k_2} \Omega_{k_0 k_1 k_2} T_{i_0}^{k_0}(z) \, T_{i_1}^{k_1}\left(z + \tfrac{2}{3} \hbar \sh^\vee\right) &T_{i_2}^{k_2}\left(z + \tfrac{4}{3} \hbar \sh^\vee\right) = \Omega_{i_0 i_1 i_2} \;.
\end{align}
Here $R^{ik}_{rs}$ is the $R$-matrix associated to the crossing of two Wilson lines in the representation $\mathbf{26}$ of $\ffour$, and $\sh^\vee = 9$ is the dual Coxeter number of $\ffour$. 

In the classical limit $\hbar \to 0$, these relations describe the universal enveloping algebra of $\mf{f}_4[[z]]$.

For $\esix$, almost the same story holds, except that the $\mathbf{27}$ of $\esix$ is not self-dual.  Therefore we only have two relations, instead of three.  If we choose a basis for the $\mathbf{27}$ and let $\Omega_{ijk}$ denote the invariant tensor in $\mathbf{27}^{\otimes 3}$, then the Yangian for $\esix$ is generated, as always, by the coefficients of a series \eqref{T_series}
where $i,j$ run from $1$ to $27$.  The relations are
\begin{align}
&\sum_{r,s} R^{i k}_{r s}(z - z') T^r_j(z') T^s_l(z) = \sum_{r,s} T^i_r(z) T^k_s(z')R^{rs}_{jl}(z -z') \;,\\
&\sum \Omega_{k_0 k_1 k_2} T_{i_0}^{k_0}\left(z\right) T_{i_1}^{k_1}\left(z + \tfrac{2}{3} \hbar \sh^\vee \right) T_{i_2}^{k_2}\left(z + \tfrac{4}{3} \hbar \sh^\vee\right) = \Omega_{i_0 i_1 i_2} \;.
\end{align}
Since $\mf{e}_6$ consists of the endomorphisms of the representation $\mathbf{27}$ preserving the cubic invariant tensor, these relations describe the universal enveloping algebra of $\mf{e}_6[[z]]$ in the classical limit.

\subsection{\texorpdfstring{$Y(\eseven)$}{Y(e(7))}}
Classically,  $\eseven$ is the  Lie algebra of endomorphisms of its representation $\mathbf{56}$ which preserve two invariant tensors.  The first is a skew-symmetric bilinear form, and the second is a totally symmetric quartic tensor.  

In \cite{Part1}, section 7.11, we showed that the $\mathbf{56}$ of $\eseven$ defines a Wilson line at the quantum level. We found that there was a unique, up to scale, invariant tensor in $\mathbf{56}^{\otimes 4}$ which is invariant under the dihedral group $D_4$ and which defines a vertex linking four copies of the Wilson line in the $\mathbf{56}$. The angles between the four Wilson lines at the vertex are all $\pi/2$.   

Let us denote this quartic tensor by $\Omega_{ijkl}$, and the skew-symmetric bilinear form by $\omega_{ij}$.  Then by the usual reasoning, a presentation of the Yangian is provided by the coefficients of the series $T^i_j(z)$, where the indices $i,j$ run from $1$ to $56$, subject to the relations  
\begin{align}
&\sum_{r,s} R^{i k}_{r s}(z - z') T^r_j(z') T^s_l(z) = \sum_{r,s} T^i_r(z) T^k_s(z')R^{rs}_{jl}(z -z') \;,\\
	&\sum \Omega_{k_0 k_1 k_2 k_3} T_{i_0}^{k_0}\left(z\right) T_{i_1}^{k_1}\left(z + \tfrac{1}{2} \hbar \sh^\vee \right) T_{i_2}^{k_2}\left(z + \tfrac{2}{2} \hbar \sh^\vee\right) T_{i_3}^{k_3} \left(z + \tfrac{3}{2} \hbar \sh^\vee  \right) = \Omega_{i_0 i_1 i_2 i_3}\;, \\
&\sum T^k_i\left(z\right) T^l_j\left(z + \hbar \sh^\vee\right)  \omega_{kl}  = \omega_{ij} \label{equation_rtt_e7} \;.
\end{align}
The dual Coxeter number of $\eseven$ is $\sh^\vee = 18$.  

 As explained in \cite{Part1}, section 7.11, the
quantizable four-valent vertex $\Omega_{ijkl}$ of $\eseven$ is a linear combination of the essentially unique completely symmetric invariant in $\mathrm{Sym}^4\mathbf{56}$ with
a correction proportional to $\omega_{ij}\omega_{kl}+\omega_{jk}\omega_{li}$ that is needed to cancel an anomaly.   This latter tensor is not completely symmetric, but it does
have the $D_4$ symmetry of a configuration of four rays meeting in the plane at relative angles $\pi/2$.

\section{\texorpdfstring{Uniqueness of the Rational $R$-matrix}{Uniqueness of the Rational R-matrix}}\label{unirational}

\subsection{Overview}

The starting point in deriving RTT presentations of the Yangian was a presumed knowledge of the $R$-matrix.  We know from an explicit lowest order
computation in \cite{Part1} that the four-dimensional gauge theory 
generates in the lowest nontrivial order, for the case $C=\C$,  the
 standard quasiclassical $r$-matrix $c/z$, where $c=t_a\otimes t_a$.  
 We would like to know whether the full quantum $R$-matrix that comes from the gauge theory  agrees with the standard rational
 $R$-matrix as studied in the integrable systems literature.   Of course, this must be so in order for the algebras whose RTT presentations we have deduced from the gauge theory
  to agree with the usual Yangian algebras.
 
 A direct computation of the full quantum $R$-matrix from the quantum field theory would be prohibitively difficult, as it requires the computation of Feynman diagrams at all loops.  To match with the $R$-matrix appearing in the literature, we need to constrain the field theory $R$-matrix using the formal properties it satisfies.

There are two ways one might approach this problem. A somewhat abstract approach was pursued in \cite{Costello_2013}. There, it was argued that formal properties of the category of line operators are strong enough to make this category equivalent to the category of representations of the Yangian, and to match the $R$-matrix coming from field theory with that coming from the Yangian.  This relied heavily on a uniqueness result of Drinfeld \cite{Drinfeld_ICM}.

In this section we will develop a more direct way to show that the field theory $R$-matrix matches with the standard one. The properties of the field theory $R$-matrix we will use are ones we have already discussed extensively: the Yang-Baxter equation, and the supplementary equations such as the quantum determinant relation that appeared in our study of the RTT presentation.  We will show that for the lowest-dimensional representation of a simple Lie algebra which is not $\mf{e}_8$, these relations are strong enough to uniquely constrain the quantum $R$-matrix.   

Let us start with the example of the fundamental representation of $\mf{sl}_N$.  Suppose we have two different $R$-matrices 
\begin{equation}
R(z),R'(z) \in \mf{gl}_N \otimes \mf{gl}_N [[\hbar]] \;,
\end{equation}   
which are rational functions of the spectral parameter $z$ and formal power
series in $\hbar$. Suppose that each satisfies the quantum Yang-Baxter equation, and that each is a series in $\hbar$ of the form
\begin{equation}
\op{Id} \otimes \op{Id} + \frac{\hbar c}{ z} + \O(\hbar^2) \;.
\end{equation}
 We further ask that $R,R'$ are compatible with the natural symmetries: both are $\sl_N$-invariant and both are invariant under the symmetry which scales both $\hbar$ and $z$.    

We will show that if $R,R'$ both satisfy the quantum determinant equation, then they are the same. 

Let us describe the quantum determinant equation for $R, R'$.  As we have seen, given a horizontal Wilson line $W$ and a vertical Wilson line in the fundamental representation, we get an operator $T^i_j(z) : W \to W$.  In the case that the horizontal Wilson line is also in the fundamental representation, and is placed at $z = 0$, then the operator $T^i_j(z) \in \mf{gl}_N$ is nothing but the $R$-matrix for two copies of the fundamental representation (up to a shift in $z$ coming from the rotation of the horizontal Wilson line to vertical).

The operators $T^i_j(z)$ satisfy the quantum determinant equation
\begin{equation}
\sum_{k_r} \op{Alt}(k_0,\dots,k_{n-1}) T^{k_0}_{0}(z ) T^{k_1}_{1}(z + 2 \hbar  )\cdots T^{k_{n-1}}_{n-1}(z + 2(n-1) \hbar ) = 1 \;.
\end{equation}
In this equation the composition is taken in $\op{End}(W) = \mf{gl}_N$.  This equation is therefore a constraint on the $R$-matrix.

Suppose we have two $R$-matrices $R,R'$ which both satisfy the quantum determinant equation, as well as the other properties listed above.  Suppose that they agree modulo $\hbar^n$, and let us write
\begin{equation}
R' (z) = R(z) + \hbar^n r^{(n)}(z) + \hbar^{n+1}  r^{(n+1)}(z) + O(\hbar^{n+2}) \;.
\end{equation}
Since $R(z)$ satisfies the quantum determinant equation, imposing this equation for $R'(z)$ tells us that 
\begin{equation}
r^{(n),j}_i(z) \in \mf{sl}_N  \;,
\end{equation}
where we view $r^{(n),j}_i(z)$ as an modification of $T^i_j(z)$ at order $\hbar^n$.  Thus, $r^{(n)}(z) \in \mf{gl}_N \otimes \mf{sl}_N$.  Reversing the roles
of the two fundamental Wilson lines in this  argument tells us that actually $r^{(n)}(z) \in \mf{sl}_N \otimes \mf{sl}_N$. 

Next, $\sl_N$ invariance together with invariance under the symmetry which scales $\hbar$ and $z$ simultaneously tells us that
\begin{equation}
r^{(n)}(z) = \lambda z^{-n} c \;,
\end{equation}
where $c \in \mf{sl}_N \otimes \mf{sl}_N$ is the quadratic Casimir and $\lambda \in \C$ is a constant.

We need to show that $r^{(n)}(z) = 0$.  We will use the Yang-Baxter equation to show this.

The Yang-Baxter equation for $R'(z)$ state that
\begin{equation}
R'_{12}(z_1 - z_2)R'_{13}(z_1 - z_3)R'_{23}(z_2 - z_3) = R'_{23}(z_2 - z_3) R'_{13}(z_1 - z_3) R'_{12}(z_1 - z_2) \;.
\end{equation}
Let us analyze the coefficient of $\hbar^{n+1}$ in this equation.  We will use the facts that $R' = R + \hbar^n r^{(n)} + \hbar^{n+1} r^{(n+1)}+\dots$, that $R$ also satisfies the Yang-Baxter equation, and that the coefficient of  $\hbar$ in $R$ is $c z^{-1}$.  The coefficient of $\hbar^{n+1}$ in the Yang-Baxter equation for $R'$ does not impose any constraints on $r^{(n+1)}$. However it does show that $r^{(n)}$ satisfies the equation 
\begin{align}
\begin{split}
&\left[\frac{\hbar c_{12}}{z_1 - z_2} + \hbar^n r^{(n)}_{12}(z_1 - z_2),  \frac{\hbar c_{13}}{z_1 - z_3} + \hbar^n r^{(n)}_{13}(z_1 - z_3)\right] \\ 
& + \left[\frac{\hbar c_{13}}{z_1 - z_3} + \hbar^n r^{(n)}_{13}(z_1 - z_3),  \frac{\hbar c_{23}}{z_2 - z_3} + \hbar^n r^{(n)}_{23}(z_2 - z_3) \right]\\
&+ \left[\frac{\hbar c_{12}}{z_1 - z_2} + \hbar^n r^{(n)}_{12}(z_1 - z_2),  \frac{\hbar c_{23}}{z_2 - z_3} + \hbar^n r^{(n)}_{23}(z_2 - z_3) \right] = 0 \quad (\text{modulo } \hbar^{n+2}) \;.    \label{cybe} 
\end{split}
\end{align}
This equation simply says that $c z^{-1} + \eps r^{(n)}(z)$ satisfies the classical Yang-Baxter equation modulo $\eps^2$.  

Recall that $r^{(n)}(z) = \lambda c z^{-n}$ for some constant $\lambda$. We will show that equation \eqref{cybe} implies $\lambda = 0$.  One way to see this is to use Belavin-Drinfeld's classification of solutions of the classical Yang-Baxter equation.  They show that every solution of the classical Yang-Baxer equation (CYBE) in $\mf{g} \otimes \mf{g}$ for a simple Lie algebra $\mf{g}$ must have only first order poles, so that in particular we cannot deform the solution $cz^{-1}$ into a new solution by adding on a term with a higher-order pole.  

However, we can verify easily by hand that equation \ref{cybe} implies $r^{(n)}(z) = 0$.   This equation contains $6$ terms, acted on simply transitively by the symmetric group on three letters. Each term is a product of an element of $\mf{sl}_N \otimes \mf{sl}_N \otimes \mf{sl}_N$ multiplied by the function $(z_{\sigma(1)} - z_{\sigma(2)})^{-1}(z_{\sigma(2)} - z_{\sigma(3)})^{-n}$ for some $\sigma \in S_3$.  These six rational functions of three variables are linearly independent, so the coefficient for each of them must be zero.  The element in $\mf{sl}_N^{\otimes 3}$ is a permutation of $\lambda[c_{12},c_{23}]$ and these elements are all non-zero if $\lambda$ is non-zero.  Therefore, $\lambda$ must be zero, so that $r^{(n)}(z) = 0$. 

We have found, in the case of $\mf{sl}_N$, that the $R$-matrix is uniquely determined by the Yang-Baxter equation, symmetry properties, and  a constraint coming from the quantum determinant.   Accordingly, for $\sl_N$, the rational $R$-matrix coming from the gauge theory agrees with the standard one constructed in the integrable systems literature.

\subsection{\texorpdfstring{The Failure of Uniqueness of the $R$-Matrix for $\mf{gl}_N$}{The failure of uniqueness of the R-matrix for gl(N)}}

For $\mf{gl}_N$, the $R$-matrix is not uniquely constrained by the equations which it satisfies.  Without affecting the RTT relations or the Yang-Baxter equation,
we  can always multiply the $R$-matrix by an an expression of the form $\exp( f(\hbar/z))$ where $f(\hbar/z)$ is a series of the form $c_1 \hbar/z + c_2 (\hbar/z)^2 + \dots$.  The unitarity equation on the $R$-matrix forces $f$ to be an odd function of $\hbar/z$, but there are still infinitely many free parameters in the $R$-matrix.

In this section, we will see why the previous argument breaks down for $\mf{gl}_N$, and we will find the gauge-theory origin of these parameters. 

Let us analyze where the argument breaks down in the case of $\mf{gl}_N$. If we do not have the quantum determinant relation, then $r^{(k)}$ cannot be constrained to be in $\mf{sl}_N \otimes \mf{sl}_N$.  However, because of $\gl_N$ invariance and other symmetries, we find that
\begin{equation}
r^{(k)} = \lambda_1 \hbar^k z^{-k} c + \lambda_2 \hbar^k z^{-k} \op{Id} \otimes \op{Id}
\end{equation}
for two complex numbers $\lambda_1,\lambda_2$.  The Yang-Baxter equation forces $\lambda_1 = 0$.  We cannot, however, fix the ambiguity of adding on a multiple of the identity at each order in $\hbar$.  Therefore the $R$-matrix is unique up to a transformation of the form
 \begin{equation}  \label{free_transformation}
 R(z) \mapsto e^{f\left( {\hbar}/{z} \right)} R(z)\;,
 \end{equation}
 where $f$ is an arbitrary series in one variable whose leading term is zero.  Imposing the unitarity constraint $R_{21}(z) R_{12}(-z)=1$ (see \cite{Part1}, eqn. (2.3))  restricts $f$ to being an odd function.

\def\Id{{\mathrm{Id}}}
This ambiguity is visible directly in the gauge theory.   In perturbation theory around the trivial flat connection, the
 $\mf{gl}_N$ gauge theory is a product of the $\mf{sl}_N$ theory with the free $\mf{gl}_1$ theory. 
 Let us write $A_{\Id}=\frac{1}{N}\Tr\, A$ for the identity component of the gauge field $A$.  
The coupling of $A_\Id$ to  a fundamental Wilson line is far from uniquely determined. In general, to the usual coupling of $A_\Id$ in a fundamental Wilson line, we can add a coupling of the form 
\begin{equation}
	 \sum_{n \ge 0} c_n \int_\R \hbar^n \partial_z^n A_{\op{Id}} \;, \label{free_coupling} 
\end{equation}
for a sequence of constants $c_n$. 

Let $R(z) \in \mf{gl}_N \otimes \mf{gl}_N$ denote the $R$-matrix for two copies of the fundamental representation of $\gl_N$, and let $R_{\rm free}(z,c_0,c_1,\dots)$ denote the $R$-matrix for the $\gl_1$ gauge theory where a one-dimensional Wilson line is coupled by \eqn \eqref{free_coupling}.

The notation $R_{\rm free}$
reflects the fact that the $\gl_1$ theory is, in fact, free.   Then the $R$-matrix for the $\gl_N$ fundamental representation modified so that the coupling of $A_\Id$ is the usual coupling plus the expression in \eqn \eqref{free_coupling} is
 \begin{equation}
R_{\rm free}(z,c_1,c_2,\dots) R(z) \;.
\end{equation} 
The free $R$-matrix can be calculated by a Feynman diagram analysis. Since it is a calculation in a free theory, the result is just the exponential of a contribution
from single gauge boson exchange.  Contributions from a propagator going from one Wilson line to itself vanish.\footnote{\label{thisone} This follows from the fact that the propagator $\langle A_i(x,y,z,\bar z)A_j(x',y',z',\bar z')\rangle $ of this theory is antisymmetric in $i$ and $j$. See the explicit formulas in \cite{Part1}.}  The contribution from gauge
boson exchange between the two Wilson lines gives  
\begin{align}
\begin{split}
R_{\rm free}(z,c_0,c_1,\dots) &= \exp \left(\sum_{n,m \ge 0 } \hbar^{n+m+1} (-1)^m c_n c_m \partial_z^{n+m}  z^{-1}  \right)\\
	&= \exp \left( \sum_{\substack{n,m \ge  0 \\ n+m \text{ even} }} \hbar^{n+m+1} (-1)^m c_n c_m \partial_z^{n+m} z^{-1}  \right) \;,
\end{split}
\end{align}
This is an odd function of $\hbar/z$, and all odd functions of $\hbar/z$ can be constructed in this way by suitable choices of the constants $c_i$. 	
	
In this way we see that changing the way the identity component of the gauge field is coupled modifies the $R$-matrix of the fundamental representation by a transformation of the form \eqref{free_transformation}, where $f(\hbar/z)$ is an arbitrary odd function of one variable.  

Therefore \emph{every} solution to the quantum Yang-Baxter equation for the fundamental representation of $\gl_N$, which is unitary, 
$\gl_N$-invariant and compatible with the symmetry which scales $z$ and $\hbar$, arises by coupling the identity component of the gauge field to a fundamental Wilson line in a general way. 

\subsection{Other Simple Lie Algebras}

The issue that we have just discussed for $\gl_N$ does not have an analog for a simple Lie algebra.
For the other Lie algebras for which we can find an RTT presentation of the Yangian (that is, all  cases except $\eeight$), a similar argument to what we gave for $\sl_N$
shows that the $R$-matrix is uniquely fixed by various formal properties we have derived. We will phrase the statement as a proposition.
\begin{proposition}
Fix a simple Lie algebra $\mf{g}$ which is not $\eeight$, and let $V$ denote its smallest non-trivial representation.  Then,  there is a unique $R$-matrix
\begin{equation}
R(z) \in \op{End}(V) \otimes \op{End}(V)[[\hbar]] \;,
\end{equation}
which is a series in $\hbar$ whose coefficients are rational functions of $z$, and which has the following properties.
\begin{enumerate}
\item $R$ is $G$-invariant.
\item $R$ is invariant under the symmetry scaling $\hbar$ and $z$ simultaneously.
\item $R$ satisfies the Yang-Baxter equation.
\item $R$ is unitary.
\item If we view the matrix $R$ as giving an operator $T^i_j(z) : V \to V$, where $i,j$ runs over a basis of $V$, then the operators $T^i_j(z)$ satisfy the extra constraints we imposed when giving an RTT presentation of the Yangian for $\mf{g}$. 
\end{enumerate}
\end{proposition}
\begin{proof}
The proof is almost identical to the one presented in the case $\mf{g} = \mf{sl}_N$. Suppose we have two such $R$-matrices, $R(z)$ and $R'(z)$, which agree modulo $\hbar^n$ and which satisfy $R'(z) = R(z) + \hbar^n r^{(n)}(z)$ modulo $\hbar^{n+1}$.    In the case of $\mf{sl}_N$, the quantum determinant constraint showed us that $r^{(n)}(z) \in \mf{sl}_N \otimes \mf{sl}_N$.  In the general case, the extra relations we add to find an RTT presentation show us that $r^{(n)}(z) \in \mf{g} \otimes \mf{g}$. This is because these relations come from a set of invariant tensors in $V$ such that $\mf{g}$ is precisely the subalgebra of $\op{End}(V)$ preserving these tensors.

Once we know that $r^{(n)}(z) \in \mf{g} \otimes \mf{g}$, the previous argument applies.
\end{proof}


\section{RTT Presentations in the Trigonometric Case}\label{trigonometric}

\subsection{Initial Steps}\label{inisteps}

In this section we will derive the RTT presentation of the quantum loop algebra from field theory.  

We will start with the simplest case, when $\g = \gl_N$. The basic setup is very similar to what we used in the rational case in section \ref{RTT_gl}, so we will
be more brief.  We consider our four-dimensional field theory on $\C \times \C^\times$, where $\C^\times$ is the complex $z$ plane with the origin omitted. The symmetry of $\C^\times$ is multiplicative, so the
$R$-matrix $R(z,z')$ will now be a function of $z/z'$, not $z-z'$ as in the rational case.

 We place an arbitrary line operator in the $x$-direction (which as usual we regard as horizontal) labeled by some representation $W$.  We also  place a vertical Wilson line in the fundamental representation of $\mf{gl}_N$  at some point $z \in \C^\times$, labeling its ends by
  incoming and outgoing states, as in the rational case. 
 Thus the initial picture is the basic one of Fig. \ref{figure_crossing}, which served as our starting point in the rational case.   There is one essential
 difference, however.  In the rational case, the vertical Wilson line decouples from the horizontal one in, and only in, the limit $z\to\infty$.  In the trigonometric case,
 there are two decoupling limits, namely $z\to 0$ and $z\to\infty$.  Accordingly, to get a full picture, we will have to consider two separate expansions of the $R$-matrix,
 and not just one as before.

As always, the $R$-matrix with chosen initial and final states on the vertical Wilson line is an operator
\begin{equation}
T^i_j(z)  : W \to W \;.
\end{equation}
Expanding near $z=0$, we get a power series of the form
\begin{equation}
T^i_j(z) = t^i_j[0] + z t^i_j[1] + \dots  \;.
\end{equation}
Each entry in this series is a linear operator on $W$. 
 (Our conventions here are slightly different than those in the rational case, in that
we do not include explicit factors of $\hbar$ in writing this expansion. Thus $t^i_j[0]=\delta^i_j+\O(\hbar)$, while $t^i_j[k]$ is of order $\hbar$ for $k>0$.)   

It is important now to remember that in section 9 of \cite{Part1}, to describe trigonometric solutions of the Yang-Baxter equation (or merely to quantize
the underlying gauge theory on $\R^2\times \C^\times$), a slightly unusual condition was placed
on the gauge field $A$ at $z=0$ and at $z=\infty$.  The condition was that $A$ is upper-triangular at $z=0$ and lower-triangular at $z=\infty$
 (there were also some additional constraints mixing the diagonal components of $A$ with a gauge field of a second copy of the Cartan subalgebra; we will analyze these 
 constraints shortly).  Since $T^i_j[z]$ comes from a quantum averaging applied to the holonomy of $A$ on a vertical line in $\R^2$ at given $z$, the
 fact that $A$ is upper triangular at $z=0$  tells us that $t^i_j[0] = 0$ for $i < j$. 

Similarly, we can expand a vertical Wilson line near $z=0$. We write the expansion in the form
\begin{equation}
S^i_j(z) = s^i_j[0] + \frac{1}{z} s^i_j[1] + \dots  \;,
\end{equation}
where $s^i_j[0] = 0$ for $i > j$.  Here we write the  $R$-matrix for the crossing of vertical and horizontal Wilson lines
as $S^i_j(z)$ if we intend an expansion near $z=\infty$, or as $T^i_j(z)$ if we intend an expansion near $z=0$.

Crossing a pair of these  vertical Wilson lines will give rise to RTT relations.  In the rational case, there was a single RTT relation. In the trigonometric case, there are three RTT relations, coming from crossing a pair of lines near $z=0$, a pair of lines near $z=\infty$, or a pair of Wilson lines of opposite types. The relations are:
\begin{align}
\begin{split}
\sum_{r,s} R^{i k}_{r s}\left(z / z'\right) T^r_j\left(z'\right) T^s_l\left(z\right) &= \sum_{r,s} T^i_r\left(z\right) T^k_s\left(z'\right)R^{rs}_{jl}\left(z/ z'\right)\;,\\
\sum_{r,s} R^{i k}_{r s}\left(z / z'\right) S^r_j\left(z'\right) S^s_l\left(z\right) &= \sum_{r,s} S^i_r\left(z\right) S^k_s\left(z'\right)R^{rs}_{jl}\left(z/ z'\right)\;, \\ 
\sum_{r,s} R^{i k}_{r s}\left(z / z'\right) T^r_j\left(z'\right) S^s_l\left(z\right) &= \sum_{r,s} S^i_r\left(z\right) T^k_s\left(z'\right)R^{rs}_{jl}\left(z/ z' \right)\;.
\end{split}
\end{align}
As in the rational case, one can derive from these relations certain commutation relations for an algebra generated by symbols $s^i_j[n]$, $t^i_j[m]$,  where $n,m\geq 0$
and also $s^i_j[0] = 0$ if $i > j$ and $t^i_j[0] = 0$ if $i < j$.  It is known \cite{Ding-Frenkel} that these RTT relations, together with one extra relation that we will give shortly, describe the \emph{quantum loop algebra} of $\mf{gl}_N$. This is a deformation of the universal enveloping algebra of $\mf{gl}_N[z,z^{-1}]$.  

The generators of the quantum loop algebra are related to the generators $t^i_j[n]$, $s^i_j[m]$ of the RTT algebra as follows.  Let $e^i_j$ denote the elementary matrix in $\mf{gl}_N$. Then, $z^n e^i_j$ corresponds to $\hbar^{-1} s^i_j[n]$ if $n > 0$ and to $\hbar^{-1} t^i_j[n]$ if $n < 0$.  If $i < j$, $z^0 e^i_j$ corresponds to $s^i_j[0]$ and if $i > j$, $z^0 e^i_j$ corresponds to $t^i_j$.

We see, however, that there are more generators in the RTT algebra than in the quantum loop algebra.  There are two possibilities corresponding to $z^0 e^i_i$, namely $t^i_i[0]$ and $s^i_i[0]$.   It turns out that there is one more relation needed between the generators $s^i_j[n]$, $t^i_j[m]$ to find an exact match with the quantum loop algebra. The extra relation is 
\begin{equation}
s^i_i  t^i_i  = 1 \;. \label{RTT_trigonometric_final}
\end{equation}
In this relation we \emph{do not} sum over the index $i$.
 
We will show that this equation holds when the horizontal line operator satisfies an additional natural constraint. As a corollary, we will find that the Hilbert space for any line operator satisfying this constraint carries a canonical action of the quantum loop algebra. 

\subsection{Finding the Extra Relation in Field Theory}
In section 9 of \cite{Part1}, in order to define the theory on $\R^2\times \C^\times$,  we included in the gauge group a second copy of the maximal torus.  Let us denote
the Lie algebra of this second copy as $\til{\mf{h}}$, and identify it with $\C^N$.  There are additional Wilson lines associated to representations of $\til{\mf{h}}$.  Since this is an Abelian Lie algebra, we need only consider one-dimensional representations, associated to linear maps from $\til{\mf{h}}$ to $\C$. 

As before, fix a horizontal Wilson line with Hilbert space $W$.   Put a vertical Wilson line at $z \in \C^\times$ associated to the one-dimensional representation of $\til{\mf{h}}$ in which the $k^{th}$ basis vector acts by the imaginary constant $i\in\C$. This vertical Wilson line gives rise to an operator 
$$
C_k(z) : W \to W \;.
$$
\begin{definition}
We say a horizontal Wilson line is \emph{admissible} if the operator $C_k(z)$ has no poles for $z \in \mbb{P}^1$.  
\end{definition}
In perturbation theory, {\it a priori} such poles would be at $z=0$ or $z=\infty$.

Let $A^i_j$ denote the components of the gauge field of our theory which live in $\mf{gl}_N$, and let $\til{A}_i$ denote the components which live in  $\til{\mf{h}}$,
the second copy of the Cartan subalgebra.

Any Wilson line which is built only from the $A^i_j$ fields is admissible.  This can be seen by an analysis of the Feynman diagrams.   The Wilson line from which we build the operator $C_k(z)$ is, classically, defined by the exponential of the integral of $i \til{A}_k$ along a line.  The cubic interaction of the theory only involves the $A^i_j$ components of the gauge field. The propagator has components connecting $\til{A}_i$ with $\til{A}_i$, $A^i_j$ with $A^j_i$, and a mixed term connecting $A^i_i$ with $\til{A}_i$.  

The mixed term is only present because the boundary conditions we use at $z = 0$ and $z = \infty$ involve setting a linear combination of $A^i_i$ and $\til{A}_i$ to zero.  The mixed term in the propagator is therefore entirely an IR effect, and does not introduce any UV singularities. More formally,  this mixed term is a smooth two-form on
the product   $(\R^2 \times\C^\times)^2$ of two copies of space-time.  

A Feynman diagram analysis then makes it clear that if we have a horizontal Wilson line which only depends on the $A^i_j$ fields, and a vertical line which only depends on the $\til{A}_i$ fields, there are no singularities when they meet in the $z$-plane. 

Let us show that for any horizontal admissible Wilson line $W$, the operators $T^i_j(z)$, $S^i_j(z)$ satisfy the relation (\ref{RTT_trigonometric_final}).  The boundary conditions for the theory are that   
\begin{align} 
\begin{split}
A^k_k + i \til{A}_k &= 0 \text{ at } z = 0\;, \\
A^k_k - i \til{A}_k &= 0 \text{ at } z = \infty \;,\\
A^i_j &= 0 \text{ at } z = 0 ~\text { if } i < j  \;,\\
A^i_j &= 0 \text{ at } z = \infty \text { if } i > j \;.
\end{split}
\end{align}
The operator $T^k_k(z = 0)$ only depends on the $A^k_k$ component of the gauge field at $z = 0$.  This is because $T^k_k(z = 0)$ is defined as a matrix element of
the path-ordered exponential of the gauge field at $z = 0$, which is upper-triangular.  Triangularity ensures that $T^k_k(z = 0)$ only depends on the $A^k_k$ component of the gauge field.  

The boundary conditions at $z = 0$ then tell us that $T^k_k(z = 0)$ can be computed in terms of a path-ordered exponential of $-i \til{A}_k$. The operator $C_k(z = 0)$ is defined as a path-ordered exponential of $i \til{A}_k$.  We thus arrive at the equation
\begin{equation}
T^k_k(z = 0) = C_k(z = 0)^{-1} : W \to W \;,
\end{equation}
which holds where $W$ is any horizontal Wilson line.

Similarly, we have
\begin{equation}
S^k_k(z = \infty) = C_k( z = \infty) : W \to W \;.
\end{equation}
If $W$ is admissible, then $C_k(z)$ has no poles and so is constant in $z$.  This leads to the equation
\begin{equation}
S^k_k(z = \infty) T^k_k(z = 0) = 1\;,
\end{equation}
which is the remaining equation defining the quantum loop algebra. 

\subsection{Coproduct in the Trigonometric Cases}
In this section we will describe the coproduct on the quantum loop algebra, using the same method we used to describe the coproduct of the Yangian.
 
As in the rational case, the coproduct tells us how the generators $t^i_j[n]$ act on a horizontal Wilson line which is obtained by fusing two parallel horizontal Wilson line.

To understand this, let us consider, just as in the corresponding discussion in the rational case,
 a configuration of  two arbitrary horizontal line operators and one vertical line operator in the fundamental representation, with incoming and outgoing states on the vertical line.  
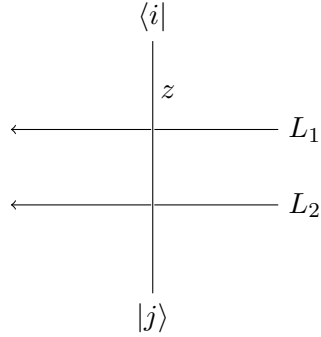
\begin{figure}[htbp]
\begin{center}
\begin{tikzpicture}
\node(In) at (0,4) {$\langle i|$};
\node(Out) at (0,0) {$|j \rangle$};
\node(HLeft) at (-2,2.5) {};
\node(HRight) at (2,2.5) {$L_1$};
\node(HLeft2) at (-2,1.5) {};
\node(HRight2) at (2,1.5) {$L_2$};
\draw[<-] (HLeft) to (HRight);
 \draw[<-] (HLeft2) to (HRight2); 
\draw[very thick, color=white] (In) to (Out);
\draw (In) to (Out);  
\node at (0.2,3) {$z$};
\end{tikzpicture}
\caption{\small{The configuration used to determine the coproduct.}}
\label{figure_coproduct_trig}
\end{center}
\end{figure}
We let $T^i_j(z, L_1 \otimes L_2)$, $S^i_j(z,L_1 \otimes L_2)$ denote the action of the operators on the composite line operator obtained from fusing the two horizontal lines.   We let $T^i_j(z,L_k)$, $S^i_j(z,L_k)$, $k=1,2$  denote the corresponding operators on the individual Wilson lines. As above, $T(z)$ is the operator we get when the vertical Wilson line is near zero, and $S(z)$ is the operator we get when the vertical Wilson line is near infinity.  We let $\mc{H}_{L_i}$, $i=1,2$ and $\mc{H}_{L_1 \otimes L_2}$ denote the Hilbert spaces at the end of the individual Wilson lines and at the end of the fused Wilson line.  Thus
\begin{equation}
\mc{H}_{L_1 \otimes L_2} = \mc{H}_{L_1} \otimes \mc{H}_{L_2} \;.
\end{equation}
The operators $T^i_j(z,L_k)$, $S^i_j(z,L_k)$ are operators on the space $\mc{H}_{L_k}$, and the operators  $T^i_j(z,L_1 \otimes L_2)$, $S^i_j(z,L_1 \otimes L_2)$  are linear operators on the spaces $\mc{H}_{L_i}$, $\mc{H}_{L_1  \otimes L_2}$ respectively.

Consider  \fig \ref{figure_coproduct_trig}. As in the discussion in the rational case,  we can move the vertical position of the two horizontal Wilson lines without changing anything (as long as the lines do not cross). When the lines are close together, the operator described by \fig \ref{figure_coproduct_trig} is  $T^i_j(z,L_1 \otimes L_2)$.  When the lines are far apart, we can decompose the operator by summing over intermediate states placed on the vertical segment between the two horizontal lines. This yields $\sum T^i_k(z,L_1)\otimes T^k_j(z,L_2)$.  We conclude that
\begin{equation}
T^i_j(z,L_1 \otimes L_2) = \sum T^i_k(z,L_1) \otimes T^k_j(z,L_2)\;,
\end{equation}
and similarly
\begin{equation}
S^i_j(z,L_1 \otimes L_2) = \sum S^i_k(z,L_1) \otimes S^k_j(z,L_2)\;.
\end{equation}
In the language of algebra, this identity tells us the coproduct on the quantum loop algebra is
\begin{align}
\begin{split}
\tr T^i_j(z) &= \sum T^i_k(z) T^k_j(z)\;,\\
\tr S^i_j(z) &= \sum S^i_k(z) S^k_j(z)\;.
\end{split}
\end{align}  
In terms of the expansion
\begin{align}
\label{TS_expansion}
\begin{split}
T^i_j(z)  &= t^i_j[0] + z t^i_j[1] + \dots\;,\\ 
S^i_j(z)  &= s^i_j[0] + \frac{1}{z} s^i_j[1] + \dots \;,
\end{split}
\end{align}
where $t^i_j[0] = 0$ for $i < j$, $s^i_j[0] = 0$ for $i > j$, and $s^i_i[0] t^i_i[0] = 1$,
the coproduct is
\begin{align}
\label{equation_coproduct_trigonometric}
\begin{split}
\tr t^i_j[n] &= \sum_{r+s=n} t^i_k[r]\otimes t^k_j[s]\;,\\ 
\tr s^i_j[n] &= \sum_{r+s=n} s^i_k[r] \otimes s^k_j[s]\;.
\end{split}
\end{align}
This is a known presentation of the coproduct on the quantum loop algebra \cite{Ding-Frenkel}. 
Note that in particular 
\begin{align}
\begin{split}
\tr t^i_i[0] &=  t^i_i[0] \otimes t^i_i[0]\;, \\
 \tr s^i_i[0] &=  s^i_i[0] \otimes s^i_i[0] \;. 
 \end{split}
\end{align}
This is consistent with the relation $s^i_i[0] t^i_i[0] = 1$ and the fact that the coproduct must be an algebra homomorphism.
  
 This concludes our discussion of $\gl_N$.  We consider $\sl_N$ next. 
  
\subsection{Quantum Determinant}

An analysis identical to that in section \ref{RTT_sl} tells us that, for the case of $\sl_N$, the following additional relations hold in the RTT algebra:
\begin{equation}
\begin{split}
\sum_{k_r} \op{Alt}\left(k_0,\dots,k_{n-1}\right) T^{k_0}_{0}\left(z \right) T^{k_1}_{1}\left(z e^{ 2 \hbar} \right)\cdots T^{k_{n-1}}_{n-1} \left(z e^{2(n-1) \hbar}  \right) = 1 \;, \\
\sum_{k_r} \op{Alt}\left(k_0,\dots,k_{n-1}\right) S^{k_0}_{0}\left(z \right) S^{k_1}_{1}\left(z e^{ 2 \hbar} \right)\cdots S^{k_{n-1}}_{n-1} \left(z e^{2(n-1) \hbar}  \right) = 1 \;.
\end{split}
\label{equation_qdet_trig} 
\end{equation}
This is deduced by moving a horizontal Wilson line above or below an $n$-fold vertex among vertical Wilson lines; the vertical Wilson lines are chosen to lie near $z=0$ or near $z=\infty$.

The only difference between this equation and the corresponding one in the rational case (\eqn \eqref{equation_qdet}) is that here we have $z e^{2 k \hbar}$ instead of $z + 2 k \hbar$ appearing in the $(k+1)$th term in the product.
As in section \ref{RTT_sl}, the shift in $z$ results from the framing anomaly.  In the multiplicative case, because the symmetry of $\C^\times $ is multiplicative,
 the framing anomaly introduces a multiplicative shift in $z$, instead of an additive one.  

The framing anomaly is a local quantity which does not depend on global features of the theory, such as boundary conditions.   Locally, we can change coordinates by setting $u = \log z$. This transforms the theory in the trigonometric setting, with action $\int(\d z/ z) \wedge \textrm{CS}(A)$, to the theory in the rational setting, with action $\int \d u \wedge \textrm{CS}(A)$.  Because the framing anomaly is local, we can compute it locally in the coordinate $u$, where we find the framing anomaly for the rational case which involves an additive shift in $u$. Since $z = \exp(u)$, the shift in $z$ will be multiplicative.

This quantum determinant relation removes the extra generators of the quantum loop algebra of $\mf{gl}_N$ that are not needed in the quantum loop algebra of $\mf{sl}_N$.

Let us see how, classically, this relation tells us that we find the generators of the universal enveloping algebra of\footnote{This is the algebra of $\sl_N$-valued functions of $z$
that are allowed to have poles at $z=0$ and at $z=\infty$.}  $\mf{sl}_N((z))$.  Let us choose slightly different generators by setting
\begin{align}
\label{til_t_def}
\begin{split}
t^i_j[n] &= \delta_{n = 0} \delta_{ij} + \hbar \,\til{t}^i_j[n] \;,\\
s^i_j[n] &= \delta_{n = 0} \delta_{ij} + \hbar \, \til{s}^i_j[n]\;.
\end{split}
\end{align}
This is a reasonable thing to do because, modulo $\hbar$, the only generators which act non-trivially on the states at the end of a horizontal Wilson line are $t^i_i[0]$ and $s^i_i[0]$, and these act by the identity. 

In addition, the fact that $s^i_i[0] t^i_i[0] = 1$ tells us that
\begin{equation}
\til{s}^i_i[0] + \til{t}^i_i[0] = 0 \;.
\end{equation} 
Finally, modulo $\hbar$, the quantum determinant relation tells us that
\begin{align}
\begin{split}
\sum t^i_i[n] &= 0 \;,\\
\sum s^i_i[n] &= 0\;.
\end{split}
\end{align}
Thus, we find the generators of $\mf{sl}_N((z))$.

\subsection{Quantum Loop Algebra for Other Lie Algebras}
In the rational case, we went beyond $\gl_N$ and $\sl_N$ to give an RTT presentation of the Yangian for all simple Lie algebras except $\eeight$.  In what follows, we will  do this
in the trigonometric case.   This requires a little more sophistication with the group theory than we have needed up to this point. 

We first need to understand
 the boundary conditions for the operators $T(z)$ and $S(z)$ in the general case. Thus, fix a simple Lie algebra $\mf{g}$ with a representation $V$ of $\mf{g}$ which quantizes to a representation of the Yangian.    Choose  a generic element $\lambda$ in the real Cartan of $\mf{g}$.  The  choice of such a $\lambda$ gives rise to an order on the set of weights, where $w > w'$ if $\lambda(w) > \lambda(w')$.   

Choose a basis of $V$ where every basis element is in a weight space.
Given a basis element $v_i$ of $V$,  we let $w_i : \mf{h} \to \C$ be the corresponding weight. 

Let us decompose $\mf{g} = \mf{n}_+ \oplus \mf{h} \oplus \mf{n}_-$ where $\mf{n}_{\pm}$ are the subspaces spanned by the positive/negative eigenspaces of $\lambda$.  Let us choose boundary conditions at $z = 0$ and $z = \infty$ based on this decomposition.  

Fix a horizontal Wilson line in a representation $W$.  Putting a vertical Wilson line at $z$ near $0$ in the representation $V$, with incoming and outgoing states $i$ and $j$, gives an operator
\begin{equation}
T^i_j(z) : W \to W \;.
\end{equation}
We need to understand what constraints the boundary conditions impose on these operators at $z = 0$.  

At $z = 0$, the gauge field only has components in $\mf{n}_+ \oplus \mf{h} \oplus \til{\mf{h}}$ where $\til{\mf{h}}$ is the second copy of the Cartan.  We will assume our representation $V$ is acted on trivially by $\til{\mf{h}}$.  The product of operators in the subalgebra $\mf{n}_+ \oplus \mf{h}$ act on the representation $V$ by matrices $A$ where $A_{ij} = 0$ if $w_i < w_j$. This implies that  $T^i_j (z = 0) = 0$ when $w_i < w_j$.  

Similarly, if $S(z)$ is the operator coming from a vertical Wilson line at $z$ near infinity, we have $S^i_j (z = \infty) = 0$ when $w_j < w_i$. 

Next, we need to understand the further relations coming from the boundary conditions on the components of the gauge field which lie in the two copies of the Cartan. 

For each weight $w$, let us define a one-dimensional representation of $\mf{g} \oplus \til{\mf{h}}$ where $\mf{g}$ acts trivially and $\til{\mf{h}}$ acts on $\C$ by $i w$.  We let 
\begin{equation}
C_w(z) : W \to W
\end{equation} 
be the operator coming from putting a vertical Wilson line in this one-dimensional representation crossing the horizontal Wilson line in the representation $W$.

 We will assume that the horizontal Wilson line is \emph{admissible}, meaning that $C_w(z)$ for each $w$ is constant as a function of $z$. As in our discussion of $\gl_N$, a sufficient condition for this is that 
 the horizontal Wilson line is not coupled to the components of the gauge field in $\til{\mf{h}}$.      

Now we can discuss the behavior of $T^i_j(z=0)$ for the case that basis vectors $v_i$ and $v_j$ have equal weights, say $w_i=w_j=w$.   This can only depend on the $\mf{h}$ valued
part of $A$ at $z=0$, and on the gauge field $\tilde A$ valued in the second copy $\tilde {\mf{h}}$ of the Cartan, because the $\mf{n}_+$-valued part of $A$ acts 
 in such a way as to strictly increase the weights.
 For each weight $w$, let $A_w$ be the one-form field obtained by applying the linear function $w: \mf{h} \to \C$ to the $\mf{h}$-valued part of the gauge field. Similarly define $\til{A}_w$ to be the one-form field where we apply $w$ to the components of the gauge field living in the second copy of the Cartan $\til{\mf{h}}$.  Then, our boundary conditions tell us that \begin{align}
\begin{split}\label{dolbo}
A_w + i \til{A}_w &= 0 \text{ at } z = 0 \;,\\
A_w - i \til{A}_w &= 0 \text{ at } z = \infty \;.
\end{split}
\end{align}
Let us now consider the restriction of $T^i_j(z=0)$ to the case that $w_i=w_j=w$.  In this space, $A$ acts via $A_w$, and the boundary condition says that at $z=0$, this is the same
as $-i\tilde A_w$.    This means that $T^i_j(z=0)$, for initial and final states of weight $w$, is the same as $\delta^i_j C_w(z=0)^{-1}$.  Here the $\delta^i_j$ comes because the action of $C_w$ depends only
on the weight; in other words, $C_w$ is a multiple of the identity in acting on states of weight $w$.  Thus
\begin{equation}\label{reback}
\left.T^i_j(z = 0)\right|_{w_i=w_j=w} =\delta^i_j  C_{w}(z = 0)^{-1} 
\end{equation} 
as operators mapping $W \to W$. Similarly,
\begin{equation}\label{eback}
\left. S^i_j(z = \infty)\right|_{ w_i=w_j=w}=\delta^i_j C_{w}(z = \infty) \;.
\end{equation}
The operator $C_w(z)$ satisfies
\begin{equation}
C_{w + w'}(z) = C_w(z) C_{w'}(z) \;.
\end{equation}
This statement follows from the statement that the fusion of the Wilson line in the representation of $\til{\mf{h}}$ labelled by a weight $w$ with that labelled by $w'$ is the Wilson line in the representation labelled by $w+w'$.  To verify this statement, we note that this is a microscopic statement which for $z$ far away from $0,\infty$ does not involve the boundary conditions. Since the gauge fields in $\mf{g}$ and those in $\til{\mf{h}}$ are only related to each other by the boundary conditions, this statement can be checked in the free gauge theory with gauge Lie algebra $\mf{h}$.  The absence of an interaction means that parallel Wilson lines in the representations $w,w'$ cannot exchange any gluons,\footnote{See footnote \ref{thisone}.}
 so that the classical fusion operation receives no quantum corrections.

Our horizontal Wilson line $W$ is assumed to be admissible, so the operator $C_w(z)$  is independent of $z$.  We therefore find that $C_{w'+w} = C_w(z) C_{w'}(z)$ for all values of $z$, including $z = 0$ and $\infty$. We can also just write $C_w$ instead of $C_w(z)$, because the operator is independent of $z$. 

Since $T^i_i(z = 0) = C_{w_i}^{-1}$, and $S^i_j(z = \infty) = C_{w_i}$, we find that 
\begin{enumerate} 
\item The operators $T^i_i(z = 0)$, $S^j_j (z = \infty)$ all commute with each other.  
\item If we form a basis of the weight lattice given by $\alpha_1,\dots,\alpha_r$, and we write $w_i = \sum \lambda_{ij} \alpha_j$ for integers $\lambda_{ij}$, then
\begin{align}
\begin{split}
T^i_i (z = 0) &= \prod_{j = 1}^{r} C_{\alpha_j}^{-\lambda_{ij}}\;,\\
S^i_i (z = \infty) &= \prod_{j = 1}^{r} C_{\alpha_j}^{\lambda_{ij}}\;.
\end{split}
\end{align} 
\end{enumerate} 
In particular, 
\begin{equation}
T^i_i(z = 0) S^i_i(z = \infty) = 1 \;.
\end{equation}

If $w_i=w_j$ but $i\not=j$, we can make similar statements using (\ref{reback}) and (\ref{eback}).  Actually, the only case of this we will need is for $w_i=w_j=0$.  In this case, $C_w=1$ so
\be\label{zobb}\left. T^i_j(z=0)\right|_{w_i=w_j=0}=\delta^i_j=\left. S^i_j(z=\infty)\right|_{w_i=w_j=0}.\ee   The reason that this is the only case we really need is that in the representations that we will actually
use for vertical Wilson lines, the weight spaces are 1-dimensional, except that the $\mathbf{26}$ of $\ffour$ has a 2-dimensional subspace with weight 0.

\subsection{\texorpdfstring{Quantum Loop Algebra for $\mf{so}_N$ and $\mf{sp}_{2N}$}{Quantum loop algebra for so(N)and sp(2N)}}
Now take $V$ to be the vector representation of $\mf{so}_N$ or $\mf{sp}_{2N}$.  Choose a basis of weight vectors of $V$, and let $\omega_{ij}$ denote the symmetric or antisymmetric pairing in this basis.  Note that $\omega_{ij}$ is only non-zero if the weights $w_i$ corresponding to the basis vectors satisfy $w_i + w_j = 0$.   We can normalize the basis so that $\omega_{ij} = 1$ if $w_i = - w_j$ and $w_i \ge 0$.  Note that for the vector representation of $\mf{so}_N$ or $\mf{sp}_{2N}$, the weight spaces are all one-dimensional.

Following our discussion in section \ref{inisteps} above and in the rational case, we find that the RTT algebra has generators given by the coefficients of two series $S^i_j(z)$, $T^i_j(z)$ where $S^i_j(z)$ is defined near $z = \infty$ and $T^i_j(z)$ is defined near $z = 0$.  We also have generators $C_w$ for $w$ a weight of the group, which are invertible and satisfy $C_{w} C_{w'} = C_{w+w'}$.   (Mathematically, these operators form a copy of the group algebra of the weight lattice.) 

The relations that we have so far are
\begin{align}\label{RTT_trigonometric_general_first}
\begin{split}
\sum_{r,s} R^{i k}_{r s}(z / z')\,  T^r_j(z') \, T^s_l(z) &= \sum_{r,s} T^i_r(z) \, T^k_s(z')\, R^{rs}_{jl}(z/ z')\;,\\
\sum_{r,s} R^{i k}_{r s}(z / z')\,  S^r_j(z')\,  S^s_l(z) &= \sum_{r,s} S^i_r(z)\,  S^k_s(z')\, R^{rs}_{jl}(z/ z')\;,\\ 
\sum_{r,s} R^{i k}_{r s}(z / z')\,  T^r_j(z') \, S^s_l(z) &= \sum_{r,s} S^i_r(z) \, T^k_s(z')\, R^{rs}_{jl}(z/ z')\;,\\
\end{split}
\end{align}
and
\begin{align}\label{RTT_trigonometric_general_last}
\begin{split}
T^i_j (z = 0) &= 0 \text{ if } w_i < w_j\;, \\ 
 S^i_j (z = \infty) &= 0 \text{ if } w_i > w_j\;, \\  
T^i_i(z = 0)  &= C_{w_i}^{-1}\;, \\
S^i_i(z = \infty)  &= C_{w_i} \;.
\end{split}
\end{align}
These relations hold for the RTT algebra associated to any representation of any simple Lie algebra.  In addition, when the representation has a pairing, we have the relations
\begin{align} \label{RTT_pairing_trig}
\begin{split}
\sum T^k_i(z) \, T^l_j\left(z e^{\hbar \sh^\vee}\right) \omega_{kl}  &= \omega_{ij} \;,\\ 
 \sum S^k_i(z) \, S^l_j\left(z e^{\hbar \sh^\vee}\right) \omega_{kl}  &= \omega_{ij}  \;.
\end{split}
\end{align}
The last two relations come from the pairing, as we discussed in the rational case for the algebras $\mf{so}_N$, $\mf{sp}_{2N}$.  The only difference with the rational case is that we use a multiplicative instead of an additive shift in $z$.  

In the relations \eqref{RTT_trigonometric_general_last},  it is essential that we use a basis given by weight vectors. If we use instead, for example, an orthonormal basis in the
case of $\mf{so}_N$, we would find  more complicated relations. 

Let us verify that, in the limit as $\hbar \to 0$, this algebra describes the universal enveloping algebra of $\mf{so}_N[z,z^{-1}]$ or $\mf{sp}_{2N}[z,z^{-1}]$.   

Let us expand $T^i_j(z)$ and $S^i_j(z)$ as in \eqn \eqref{TS_expansion}.
Since we are interested in the limit $\hbar\to 0$, 
let us denote the leading term of the $\til{t}^i_j[n], \til{s}^i_j[n]$
by  $\dtil{t}^i_j[n], \dtil{s}^i_j[n]$, so that \eqn \eqref{til_t_def}
becomes
\begin{align}
\label{dtil_t_def}
\begin{split}
t^i_j[n] &= \delta_{n = 0} \delta_{ij} + \hbar \,\dtil{t}^i_j[n]+ \O(\hbar^2) \;,\\
s^i_j[n] &= \delta_{n = 0} \delta_{ij} + \hbar \, \dtil{s}^i_j[n]+ \O(\hbar^2)\;.
\end{split}
\end{align}
The operators $\dtil{t}^i_j[n]$, $\dtil{s}^i_j[n]$  come from the exchange of a single gluon between the fundamental Wilson line and the Wilson line in the representation $W$, and so tell us how the Wilson line $W$ is coupled to the gauge field.

Similarly, when we work modulo $\hbar$ the operator $C_w$ is the identity. This is because it comes from the exchange of gluons between a vertical Wilson line in a representation of $\til{\mf{h}}$ and a horizontal Wilson line.  We can therefore expand
\begin{equation}
C_{w} = 1 + \hbar \, c_w+ \O(\hbar^2) \;.
\end{equation}

Modulo $\hbar$, we have the following equations for the operators $\dtil{t}^i_j[n]$, $\dtil{s}^i_j[n]$, and $c_w$: 
\begin{align}
\begin{split}
\dtil{s}^i_i[0] &= c_{w_i} \;,\\
[c_w,c_{w'}] &= 0\;, \\
c_{w+w'} &= c_w + c_{w'} \;,\\
 \sum \dtil{t}^k_i[n] \delta^l_j  \omega_{kl} + \delta^k_i \dtil{t}^l_j[n] \omega_{kl}  &=0 \text{ for } n \ge 0 \;,\\
 \sum \dtil{s}^k_i[n] \delta^l_j  \omega_{kl} + \delta^k_i \dtil{s}^l_j[n] \omega_{kl} &=0 \text{ for } n \ge 0 \;.
\end{split}
\end{align}
From this we see that the operators $\dtil{t}^i_i[0]$ and $\dtil{s}^i_i[0]$ are redundant, and can be replaced by a copy of the Cartan of $\mf{g}$ spanned by the operators $c_w$.  The operators $\dtil{t}^k_i[m]$ for $m > 0$ and $\dtil{s}^k_i[m]$ for $m > 0$ preserve the pairing on $V$, and so are elements of $\mf{so}_N$ or $\mf{sp}_{2N}$.  The operators $\dtil{t}^i_j[0]$ and $\dtil{s}^i_j[0]$ also preserve the pairing, and are therefore elements of $\mf{so}_N$ or $\mf{sp}_{2N}$. 

Since $\dtil{t}^i_j[0] = 0$ if $w_i < w_j$,  we see that the operators $\dtil{t}^i_j[0]$ span a copy of $\mf{n}_+$. Similarly, the operators $\dtil{s}^i_j[0]$ span a copy of $\mf{n}_-$.  In sum, we can arrange the set of generators of the algebra into a series
\begin{equation}
\sum_{m \in \Z} \alpha[m] z^m \in \mf{g}[z,z^{-1}] \;,
\end{equation}
where we identify $\alpha[m]$ with $\dtil{t}^i_j[m]$ if $m > 0$, with $\dtil{s}^i_j[-m]$ if $m < 0$. The component of $\alpha[0]$ in $\mf{n}_+$ is identified with $\dtil{t}^i_j[0]$, the component in $\mf{n}_-$ is identified with $\dtil{s}^i_j[0]$, and the component in $\mf{h}$ with the generators $c_w$.

We thus find that the set of generators of our algebra becomes, in the classical limit, the space $\mf{g}[z,z^{-1}]$. One can further compute that the relations, in the classical limit, tell us that these generators commute according to the standard commutator on the loop algebra $\mf{g}[z,z^{-1}]$. 

\subsection{Exceptional Lie Algebras}
Given what we have done so far, it is straightforward to construct the trigonometric  RTT presentation for exceptional groups $\gtwo,\ffour,\esix,\eseven$.  As in the rational case, we will build the RTT algebra from the lowest dimensional representation of the Lie algebra $\mf{g}$. We denote the representation by $V$.  As in our presentation for the groups $\mf{so}_N$ and $\mf{sp}_{2N}$, we form an RTT algebra built from generators $C_w$, $T^i_j(z)$, $S^i_j(z)$, where $i$ runs over a basis of weight vectors of the representation $V$.  
 
In every case, the relations \eqref{RTT_trigonometric_general_first}--\eqref{RTT_trigonometric_general_last} hold.  In the one case in which the representation $V$ has a weight space of dimension greater than 1
-- the zero weight subspace of the $\mathbf{26}$ of $\ffour$ --  we have the refinement (\ref{zobb}) of the last two relations in \eqref{RTT_trigonometric_general_last}.
In the cases where the representation has a pairing (that is, for the $\mathbf{7}$ of $\gtwo$, the $\mathbf{26}$ of $\ffour$ and the $\mathbf{56}$ of $\eseven$), relations \eqref{RTT_pairing_trig} hold as well. 

In addition to these relations, there are analogs of the quantum determinant associated to suitable invariant tensors.
  For the $\mathbf{7}$ of $\gtwo$, $\mathbf{26}$ of $\ffour$ and the $\mathbf{27}$ of $\esix$, there is an invariant cubic tensor which we denote $\Omega_{ijk}$.  As in the rational case, this leads to the relations 
\begin{align}
\begin{split}
\sum \Omega_{k_0 k_1 k_2} T_{i_0}^{k_0}\left(z\right) T_{i_1}^{k_1}\left(z \,  e^{ \tfrac{2}{3} \hbar \sh^\vee}\right) T_{i_2}^{k_2}\left(z\,  e^{ \tfrac{4}{3} \hbar \sh^\vee}\right) &= \Omega_{i_0 i_1 i_2}\;, \\
 \sum \Omega_{k_0 k_1 k_2} S_{i_0}^{k_0}\left(z\right)S_{i_1}^{k_1}\left(z \,  e^{ \tfrac{2}{3} \hbar \sh^\vee}\right) S_{i_2}^{k_2}\left(z\,  e^{ \tfrac{4}{3} \hbar \sh^\vee}\right) &= \Omega_{i_0 i_1 i_2}\;.  
\end{split}
\end{align}
For the $\mathbf{56}$ of $\eseven$, there is a quartic invariant tensor $\Omega_{ijkl}$ leading to the relations
\begin{align}
\begin{split}
\sum \Omega_{k_0 k_1 k_2 k_3} T_{i_0}^{k_0}\left(z\right) T_{i_1}^{k_1}\left(z \,  e^{ \tfrac{1}{2} \hbar \sh^\vee}\right) T_{i_2}^{k_2}\left(z\,  e^{ \tfrac{2}{2} \hbar \sh^\vee}\right)  T_{i_3}^{k_3}\left(z\,  e^{ \tfrac{3}{2} \hbar \sh^\vee}\right)  &= \Omega_{i_0 i_1 i_2 i_3 }\;, \\
 \sum \Omega_{k_0 k_1 k_2 k_3} S_{i_0}^{k_0}\left(z\right) S_{i_1}^{k_1}\left(z \,  e^{ \tfrac{1}{2} \hbar \sh^\vee}\right) S_{i_2}^{k_2}\left(z\,  e^{ \tfrac{2}{2} \hbar \sh^\vee}\right)  S_{i_3}^{k_3}\left(z\,  e^{ \tfrac{3}{2} \hbar \sh^\vee}\right)  &= \Omega_{i_0 i_1 i_2 i_3 }\;. 
\end{split}
\end{align}

\subsection{Comparison With Purely Three-Dimensional Chern-Simons Theory}

\def\CS{{\mathrm{CS}}}
The space $\C^\times$ admits  a $U(1)$ symmetry group  $z\to e^{i\alpha}z$, $\alpha$ real.  This is a symmetry of the action $\int_{\Sigma\times \C^\times} \frac{\d z}{z}\CS(A)$ and it makes sense
to restrict that action to $U(1)$-invariant fields.  In the process, the partial connection $A$ of the four-dimensional theory 
on $\Sigma\times \C^\times $ (which is missing a $\d  z$ term) becomes an ordinary connection on $\Sigma\times \C^\times/U(1)
=\Sigma\times\R$.  If we add to $\C^\times$ the points $z=0,\infty$, replacing $\C^\times$ by $\mathbb{CP}^1$, then the quotient becomes a closed interval $I=\mathbb{CP}^1/U(1)$.
The four-dimensional action on $\Sigma\times \C^\times $ then reduces to an ordinary three-dimensional action $4\pi i\int_{\Sigma\times I}\CS(A)$ on the three-manifold $\Sigma\times I$, with boundary conditions
that we will discuss at the endpoints of $I$.
We have arrived at a  purely three-dimensional Chern-Simons action.  Let us see how
we can exploit this fact.

Three-dimensional Chern-Simons theory with gauge group a simple Lie group $G$ is known on various grounds to be related to the quantum group deformation of the universal enveloping algebra of $\g$.
But arguably, no existing derivation of this is nearly as direct as the explanations we have given in the present paper of the relation of the four-dimensional theory to the Yangian and the quantum loop group.
What simplified the analysis in the present paper -- and the previous paper \cite{Part1} on which we drew --  is that the four-dimensional theory is infrared-free, which immediately guarantees the existence of a local procedure to evaluate the expectation of any 
configuration of Wilson lines.  Three-dimensional Chern-Simons theory is not infrared-free -- a fact that is important in many of its interesting applications, but which tends to make it difficult to give arguments
as simple as those in the present paper.

\def\At{A_{\mf t}}

However, what we have learned here suggests an interesting perspective on the purely three-dimensional case.  We start in four dimensions with gauge group $G\times \tilde T$,
where $\tilde T$ is a second copy of the maximal torus of $G$; we write $A,\tilde A$ for the $\g$ and $\tilde {\mf{t}}$-valued gauge fields. We place the same boundary conditions at $z=0$ and $z=\infty$ that we have used throughout our analysis of the quantum loop group.  Restricting to $U(1)$-invariant
fields, we get Chern-Simons theory on $\Sigma\times I$ with boundary conditions at the endpoints of $I$ that come from the conditions placed at 0 and $\infty$ in the four-dimensional theory.
Concretely,  $A$ is valued in $\mf{t}\oplus \mf{n}_+$ at one endpoint of $I$ and in $\mf{t}\oplus{n}_-$ at the other endpoint.  If $\At$ denotes the $\mf t$-valued part of $A$, one also requires, as in eqn. (\ref{dolbo}), 
that $\At+i\tilde A=0$ at one endpoint and $\At-i\tilde A=0$ at the other.

The derivation that we have given of the quantum loop group can be adapted in this situation to a purely three-dimensional derivation, with minor differences that we comment on shortly.  
Dividing by $U(1)$ has the effect of restricting to the $z$-independent generators of the quantum loop group. The algebra that we will get in the three-dimensional derivation
will be a  deformation of the universal enveloping algebra of $\g$ 
(as opposed to $\g((z))$, which we get in the derivation that starts in four dimensions).

In the reduced picture on $\Sigma\times I$, Wilson lines are associated to representations of $\g\times {\mf{t}}$ (there is no longer a $z$-dependent extension of this, since in the three-dimensional
picture $A$ and $\tilde A$ are ordinary connections).  What we call a horizontal Wilson line is  supported on a horizontal line in $\Sigma$ times a point $p\in I$.  A vertical Wilson line is supported
on a vertical line in $\Sigma$ times a point $p'\in I$.   The boundary conditions at the ends of $I$ completely break the gauge symmetry, and as a result the theory becomes infrared-free and it makes
sense to specify initial and final states at the ends of a Wilson line.
A vertical Wilson line with initial and final states $i$ and $j$ now induces a linear transformation acting on the state space of the horizontal Wilson
line.  We call this linear transformation $S^i_j$ or $T^i_j$ depending on whether $p'$ is to the left or right of $p$.   A special case of this is a vertical Wilson line, supported at $p'\in I$,
 associated to a character of $\tilde t$ (with $\g$ acting
trivially).   We say that a horizontal Wilson line at $p\in I$ is admissible if, when it crosses a special vertical Wilson line of the kind just mentioned, the resulting linear transformation is the same whether
$p'$ is to the right or left of $p$.  In particular, Wilson lines coming from representations of $\g$ (with $\tilde t$ acting trivially) are admissible.  

All arguments that we have given can be adapted in a fairly obvious way to construct a deformation of the universal enveloping algebra of $\g$ with the property that it acts on the space of states
of any admissible Wilson line.   This deformation will have a coproduct, as we discussed in the four-dimensional situation, making it into a Hopf algebra. The construction makes it manifest that the deformation in question has $T$, but not $G$, as a group of automorphisms.  This is a feature of the standard presentations of the quantum group.   

Why has this construction not been described previously?  One reason is that the above construction does not make sense in conventional Chern-Simons theory, defined with real gauge fields, a compact gauge group, and an action
that is gauge-invariant mod $2\pi \mathbb{Z}$.  The boundary conditions at the two ends of $I$ are not consistent with the pair $(A,\tilde A)$ being real, at least not if the gauge group is compact, as it is in many applications
of Chern-Simons theory.  However, in perturbation theory this is irrelevant, and actually the boundary conditions are consistent with real $(A,\tilde A)$ if we start with the split real form of $A$. (One must give
equal and opposite levels to the two factors of $G\times \tilde T$; this relative sign makes it possible to eliminate the factors of $i$ in the boundary conditions.)    Moreover,
the action in the three-dimensional reduced theory is gauge-invariant mod $2\pi \mathbb{Z}$ if the overall coefficient multiplying the
action is properly chosen.

\section{\texorpdfstring{Uniqueness of the Trigonometric $R$-matrix}{Uniqueness of the Trigonometric R-matrix}}\label{unitrigonometric}
In the rational case, we saw that the Yang-Baxter equation, together with the other relations we impose to give an RTT presentation of the Yangian, uniquely fix the $R$-matrix.  In this section we will prove a similar result in the trigonometric case.  We will find that the quantum $R$-matrix is not unique, but that all the parameters that appear have a natural explanation as parameters in the gauge theory set-up.
 
Before we turn to our analysis of the $R$-matrix, let us discuss the extra gauge-theory parameters that appear in this case.  Recall that to construct the trigonometric $R$-matrix, we studied the $4$-dimensional gauge theory with gauge group $G \times \til{H}$, where $\til{H}$ is a second copy of the Cartan.  We then chose boundary conditions at zero and infinity as follows. Let us decompose $\mf{g} = \mf{n}_+ \oplus \mf{h} \oplus \mf{n}_-$. Let 
\begin{align} 
\begin{split}
\mf{h}_0&\subset \mf{h} \oplus \til{\mf{h}}\;,\\
\mf{h}_\infty & \subset \mf{h} \oplus \til{\mf{h}}\;, 
\end{split}
\end{align}
be two complementary Lagrangian subspaces (Lagrangian with respect to the pairing which is the sum of the Killing form on $\mf{h}$ and on $\til{\mf{h}}$).  

Then, we required our gauge field to live in  
\begin{align}
\begin{split}
\mf{l}_0 &= \mf{n}_+ \oplus \mf{h}_0\;,\\
\mf{l}_\infty &= \mf{n}_+ \oplus \mf{h}_\infty\;,
\end{split}
\end{align} 
at $0$ and $\infty$ respectively.

The decomposition of $\mf{g}$ into $\mf{n}_+ \oplus \mf{h} \oplus \mf{n}_-$ is unique up to the adjoint action of $G$.  The only choice we made is that of the subspaces $\mf{h}_0$ and $\mf{h}_\infty$. 

A Lagrangian subspace in $\mf{h} \oplus \til{\mf{h}}$ is of the form
\begin{equation}
\left\{(x, i M(x) ) \mid M : \mf{h} \to \til{\mf{h}}, \ \ip{M(y), M(y')} = \ip{y,y'} \right\}\;.
\end{equation}
In other words, up to multiplication of one factor by $i$, the Lagrangian subspace is the graph of a linear isomorphism from $\mf{h}$ to $\til{\mf{h}}$ which preserves the pairing.  

If we fix an identification between $\mf{h}$ and $\til{\mf{h}}$, we see that the set of possible Lagrangian subspaces is $O(\til{\mf{h}})$, the group of isomorphisms of $\til{\mf{h}}$ preserving the pairing.  

This group acts on the field content of our theory by rotating the component of the gauge field in $\til{\mf{h}}$. The action of $O(\til{\mf{h}})$ rotates the boundary conditions at $0$ and $\infty$.  We can use this symmetry to set one of the boundary conditions, say that at $\infty$, to be the one corresponding to the identity in $O(\til{\mf{h}})$. Then the choice of boundary condition at $0$ becomes a parameter in our theory.

Let us analyze how changing the boundary condition affects the $R$-matrix.   Our previous calculation of the propagator used Lagrangians coming from the identity in $O(\til{\mf{h}})$ at $0$ and minus the identity in $O(\til{\mf{h}})$ at $\infty$.   Let us instead use a general matrix $M$ at $0$ but at infinity retain minus the identity.  For the two Lagrangians to  be transverse, the matrix $\op{Id} + M$ needs to be invertible. We let $C$ denote its inverse.    The propagator with the modified boundary conditions is 
\begin{align}  
\begin{split}
2 \pi \i \, r_M(z_1,z_2) =& \frac{1}{1 - \frac{z_1}{z_2}} \sum_\alpha X_\alpha^+ \otimes X_\alpha^-  + \frac{1}{1 - \frac{z_1}{z_2}} \sum C_{sr}  (H_r + \i M_{rk}  \til{H}_k) \otimes (H_s - \i \til{H}_s ) \;,\\ 
& - \frac{1}{1 - \frac{z_2}{z_1}} \sum_\alpha X_\alpha^- \otimes X_\alpha^+  - \frac{1}{1 - \frac{z_2}{z_1}}  \sum C_{sr} (H_s - \i \til{H}_s) \otimes (H_r + \i M_{rk} \til{H}_k )\;.
\end{split}
\end{align}
To simplify this further, we note that
\begin{align*} 
C + C^T &= (1 + M)^{-1} + (1 + M^T)^{-1} 
= (1 + M)^{-1} + (1 + M^{-1})^{-1} = 1\;. 
\end{align*}
We can therefore write
\begin{equation}
C_{rs} = \tfrac{1}{2} \delta_{rs} + C'_{rs}\;, \label{eqn_antisym}
\end{equation}
where $C'_{rs}$ is anti-symmetric.  

Inverting this procedure, we can write 
\begin{equation} 
M = \frac{1 - 2 C'}{1 + 2 C'}  \;,
 \end{equation}
where we view $C'$ as a linear operator from $\mf{h}$ to $\mf{h}$.  As long as both $1-2 C'$ and $1+2 C'$ are invertible, then $M$ is defined, invertible, and $1 + M$ is invertible.  Thus, the anti-symmetric matrices $C'$ that arise in this way are precisely those for which $\pm \tfrac{1}{2}$ are not eigenvalues.  

Using \eqref{eqn_antisym} and the fact that $C_{sr} M_{rk} = \delta_{sk} - C_{sk}$, we can further simplify the propagator to
\begin{align}
\begin{split}
 2 \pi \i \, r_M(z_1,z_2) = &  \frac{1}{1 - \frac{z_1}{z_2}} \sum_\alpha X_\alpha^+ \otimes X_\alpha^-   - \frac{1}{1 - \frac{z_2}{z_1}} \sum_\alpha X_\alpha^- \otimes X_\alpha^+ \\ 
&+  \sum \frac{z_2 + z_1}{z_2 - z_1}\tfrac{1}{2} H_r \otimes H_r 
 + \sum  H_r \otimes H_s C'_{sr} 
\\
& + \sum  \frac{z_2 + z_1}{z_2 - z_1} \tfrac{1}{2} \til{H}_r \otimes \til{H}_r - \sum C'_{rs} \til{H}_r \otimes \til{H}_s \\  
& +\i\sum \tfrac{1}{2} \delta_{rs}  \left( \til{H}_r \otimes H_s -H_s \otimes \til{H}_r\right)\\
& +\i\sum C'_{rs}  \left( \til{H}_r \otimes H_s +H_s \otimes \til{H}_r\right)\;.
\end{split}
\end{align}

We are interested in representations in which the second copy $\til{\mf{h}}$ of the Cartan acts trivially.  In such a representation, the classical $r$-matrix will be
\begin{align}
\begin{split}
2 \pi \i \, r_M(z_1,z_2) = &  \frac{1}{1 - \frac{z_1}{z_2}} \sum_\alpha X_\alpha^+ \otimes X_\alpha^-   - \frac{1}{1 - \frac{z_2}{z_1}} \sum_\alpha X_\alpha^- \otimes X_\alpha^+ \\ 
&+  \sum \frac{z_2 + z_1}{z_2 - z_1}\tfrac{1}{2} H_r \otimes H_r 
 + \sum  \left( H_r \otimes H_s \right) C'_{sr} \;. \label{trig_cybe}
\end{split}
\end{align}
In other words, in a representation where $\til{\mf{h}}$ acts trivially, we have
\begin{equation}
r_M(z_1,z_2) = r(z_1,z_2) + \frac{1}{2 \pi \i} \sum \left(H_r \otimes H_s\right) C'_{rs}\;.
\end{equation}
The effect of modifying the boundary condition is then to add an antisymmetric tensor in the Cartan $\mf{h}$ to the $r$-matrix.  This antisymmetric tensor is independent of the spectral parameter.  As we have seen above, the anti-symmetric tensors that can arise in this way are those for which $\pm \tfrac{1}{2}$ are not eigenvalues.

 So far we have analyzed what happens to the classical $R$-matrix when we change the boundary condition.  We will not give a closed-form expression for how the quantum $R$-matrix depends on the choice of boundary condition.  Instead, we analyze how the quantum $R$-matrix is affected if we change the boundary condition at order $\hbar^{k-1}$. 

The quantum $R$-matrix is built as a sum over Feynman diagrams.  Changing the boundary condition changes the propagator in the diagrams. Normally the propagator is accompanied by $\hbar$.  If we change the boundary condition at order $\hbar^{k-1}$, then we find a change in the propagator so that as well as having a coefficient of $\hbar$ it has a term with a coefficient of $\hbar^k$. This second term can be derived from our analysis above: if the matrix $M$ is $M = 1 + \hbar^{k-1} 8 \pi \i  m$, for some $m \in \mf{so}(\til{\mf{h}})$, then 
\begin{equation}
C = \tfrac{1}{2} - \hbar^{k-1} 2\pi \i m + \O(\hbar^k)  \;,
\end{equation}  
so that the new propagator is
\begin{equation}
\hbar \, r(z_1,z_2) - \hbar^{k} \sum  m_{rs}  H_r \otimes H_s\;.
\end{equation} 

Let us compute what affect this has on the quantum $R$-matrix at order $\hbar^k$.  The quantum $R$-matrix is a sum over Feynman diagrams where propagators are accompanied by $\hbar$, and vertices by $\hbar^{-1}$.  The only way the modification of the propagator can contribute to the order $\hbar^k$ term in the quantum $R$-matrix is when there is a single propagator connecting the two Wilson lines.  We therefore find that, when we change the boundary condition at order $\hbar^{k-1}$, the quantum $R$-matrix changes by
\begin{equation}
R \mapsto R - \hbar^k  \sum H_r \otimes H_s m_{rs} + \O(\hbar^{k+1}) \;.
\end{equation}
Here the anti-symmetric tensor $m_{rs}$ is arbitrary.

To sum up, we find that the parameters in the definition of the field theory affect the $R$-matrix in the following ways:
\begin{enumerate} 
\item At each order in $\hbar$, we are free to add the term $\hbar^k (\d z / z)\wedge \textrm{CS}(A)$ to the Lagrangian. (In the rational case this was forbidden by the symmetry scaling $z$ and $\hbar$).  This changes the $R$-matrix by adding the classical $R$-matrix $r^{(1)}(z)$ at order $\hbar^{k+1}$.  This change can be absorbed into a reparametrization of $\hbar$.

We can arrange such terms in a series $\sum c_k \hbar^k (\d z/ z)\wedge \textrm{CS}(A)$ where $c_k$ are the \emph{bulk coupling constants}.
\item At each order in $\hbar$,  we are free to change the boundary conditions. A change of the boundary conditions at order $\hbar^k$ adds at order $\hbar^{k+1}$ an antisymmetric tensor in $\mf{h}$ to the $R$-matrix, which is independent of the spectral parameter.  

We can arrange the choice of boundary condition in a series $\sum b_k \hbar^k$ where $b_k \in \wedge^2 \mf{h}$ and $b_0$, viewed as an operator from $\mf{h}$ to itself, does not have $\pm \tfrac{1}{2}$ as eigenvalues.   We will refer to the $b_k$ as \emph{boundary coupling constants}. 
\end{enumerate}
We will show that these parameters present in the definition of the field theory give rise to almost all possible solutions to the trigonometric $R$-matrix.

As in the rational case, any ambiguities in the solution of the quantum Yang-Baxter equation will be constrained by appealing to Belavin-Drinfeld's \cite{Belavin-Drinfeld} classification of solutions of the classical Yang-Baxter equation.  Since their classification is a little subtle in the trigonometric case, let us explain their main results.

They consider a classical $r$-matrix $r(u) \in \mf{g} \otimes \mf{g}$ which is a quasi-periodic function of a variable $u \in \C$:
\begin{equation} 
r(u + 2 \pi i) = (C \otimes 1) r(u)  \;,
 \end{equation}
where $C$ is an automorphism of the Lie algebra $\mf{g}$ of finite order. They also assume that $r(u)$ has a non-degeneracy property, which is equivalent to asking that $r(u)$ has a simple pole at $u = 0$ whose residue is the quadratic Casimir.   

They show that such solutions of the classical Yang-Baxter equation are classified by an automorphism of the Dynkin diagram of $\mf{g}$ (that is, an outer automorphism of $\mf{g}$), together with an element of the exterior square of a certain Abelian Lie algebra inside $\mf{g}$.

In the case that the automorphism of the Dynkin diagram is trivial, every such $r(u)$ is equivalent to one that is strictly periodic, $r(u + 2 \pi i ) = r(u)$.  In that case, we can view $r(u)$ as a function of $z = e^{u}$.    Then, Belavin-Drinfeld's classification tells us that there is some constant $A$ and some anti-symmetric matrix $\Gamma \in \wedge^2 \mf{h}$ such that   
\begin{equation} 
  r(z) = A\, r_{\rm standard} (z) + \Gamma_{rs} H_r \otimes H_s \;,
 \end{equation}
 with $r_{\rm standard} (z)$ some standard solution of the trigonometric CYBE. We can take for $r_{\rm standard}(z)$ the solution in \eqref{trig_cybe} with the tensor $C'$ set to zero. 
 
Note that almost all such solutions are obtained from our field theory by choicing the coupling constant and the boundary conditions appropriately.  The coupling constant can be tuned to give any value of $A$, and the classical boundary condition can be tuned to give \emph{almost} any value of $\Gamma_{rs}$. Those $\Gamma_{rs}$ that can arise are those with the property that $A^{-1} \Gamma_{rs}$, viewed as an endomorphism of $\mf{h}$, has no eigenvalues which are $\tfrac{1}{2}$ or $-\tfrac{1}{2}$. 

The remaining trigonometric solutions to the classical Yang-Baxter equation are the ones associated to a non-trivial automorphism of the Dynkin diagram $\mf{g}$. These can not be made strictly periodic in the variable $u$. Instead, in terms of the variable $z = e^u$, they can be viewed as a section of a flat bundle with fibre $\mf{g} \otimes \mf{g}$ whose monodromy is the given outer automorphism of $\mf{g}$ applied to the first factor.   These solutions do not play a role in this paper, although they can presumably be engineered by considering our four-dimensional field theory with a gauge group $G$ which has monodromy on $\C^\times$ given by an outer automorphism.   

Now let us state and prove our classification of quantum $R$-matrices in terms of field theory.
\begin{proposition}
	Let $\mf{g}$ be a simple Lie algebra which is not $\eeight$, and let $V$ be its smallest-dimensional representation. 

Consider the  solutions 
\begin{equation} 
 R(z) = 1 + \hbar\, r(z) + \dots \in \op{End}(V) \otimes \op{End}(V)[[\hbar]] 
 \end{equation}
 of the Yang-Baxter equation on $V$, satisfying the following additional properties.
\begin{enumerate}
\item  We assume $R(z)$ is, at each order in $\hbar$, a rational function of $z$ whose only poles are at $z = 1$. 
\item
We assume that $r(z)$ is a constant multiple of the $r$-matrix in \eqref{trig_cybe}, where the anti-symmetric tensor $C'_{rs}$ has eigenvalues which are not $\pm \tfrac{1}{2}$.  
\item
To $R$ we can associate an operator $T^i_j(z) : V \to V$ where $i,j$ runs over a basis of the smallest representation $V$ of $\mf{g}$.  We suppose that this operator satisfies the additional constraints we add to the RTT relation to define the Yangian. For example, for $\mf{sl}_N$, we require that this operator satisfies the quantum determinant equation
\begin{equation}
\sum_{k_r} \op{Alt}\left(k_0,\dots,k_{n-1}\right) T^{k_0}_{0} \left(z \right) T^{k_1}_{1}\left(z e^{ 2 \hbar} \right)\cdots T^{k_{n-1}}_{n-1}\left(z e^{2(n-1) \hbar}  \right) = 1 \;.
\end{equation}
\end{enumerate}

Then, there is a bijection between
\begin{enumerate}  
 \item The set of such solutions to the Yang-Baxter equation in $V$. 
\item  The possible values of the bulk and boundary coupling constants of the four-dimensional field theory on $\R^2 \times \C^\times$.  
 \end{enumerate}
\end{proposition}
\begin{proof}
The results we have explained so far show how to construct an $R$-matrix in $V$ from every choice of bulk and  boundary counter-terms.  We need to show that every possible choice of $R$-matrix arises in this way, and that the value of the bulk and boundary coupling constants is encoded in the $R$-matrix.

We will first consider the case $\mf{g} = \mf{sl}_N$.   We consider the coefficient $r(z)$ of $\hbar$ in the expansion of $R(z)$. By assumption, $r(z)$ can be engineered from some choice of boundary condition and coupling constants for the classical theory.  

Next, we assume by induction that, modulo $\hbar^k$, every $R$-matrix satisfying the constraints in the statement of the proposition arises from some choice of bulk counter-terms and of boundary condition in the field theory.   We will prove by induction that this also holds modulo $\hbar^{k+1}$. 

Thus we let $R$ be an $R$-matrix satisfying the conditions stated above, and we let $R'$ be an $R$-matrix engineered from the field theory so that $R = R'$ modulo $\hbar^k$.  Then, there is some $r^{(k)}(z)$ such that 
\begin{equation} 
R = R' + \hbar^k r^{(k)}(z) + \O(\hbar^{k+1}) \;.
 \end{equation}
The Yang-Baxter equation for $R(z)$ and $R'(z)$ implies, as in the rational case, that 
\begin{equation}
r(z) + \eps r^{(k)}(z)
\end{equation}
satisfies the CYBE modulo $\eps^2$. Further, applying the quantum determinant relation to each variable tells us that $r^{(k)}(z) \in \mf{sl}_N \otimes \mf{sl}_N$. 

Belavin-Drinfeld's classification shows that $r^{(k)}(z)$ is of the form
\begin{equation}
r^{(k)}(z) = C_k r(z) + \Gamma_k \;,
\end{equation}
where $\Gamma \in \wedge^2 \mf{h}$, and $C$ is a constant. We can change the bulk coupling constants of our four-dimensional gauge theory to absorb $C_k r(z)$ into $R'(z)$, and we can change the boundary conditions at order $\hbar^{k-1}$ to absorb $\Gamma_k$ into $R'(z)$. 

In this way we find that every possible $R$-matrix is uniquely represented as one coming from field theory, with an appropriate value of the coupling constants. 

If $\mf{g}$ is any other simple Lie algebra which is not $\eeight$, the only point at which the argument above is modified is that, instead of imposing the quantum determinant relation, we impose the constraints we add to the RTT relation to describe the quantum loop group for $\mf{g}$. (For example, for $\mf{so}_N$ or $\mf{sp}_{2N}$, we impose the constraint coming from the pairing on the fundamental representation. For $\mf{g}_2$, we impose the constraint coming from the pairing on the fundamental representation and the invariant antisymmetric $3$-tensor.)
\end{proof}

\noindent {\it Comment:}  To avoid confusion, we should perhaps explain the following.  In the present paper, we have considered only admissible Wilson lines.  This restriction leads to the family of solutions
of the classical Yang-Baxter equation, and the corresponding family of quantum $R$-matrices, that we have analyzed.  It is also possible, of course, to consider non-admissible Wilson lines, with an arbitrary weight for the
second copy of the Cartan.  This possibility has been analyzed in section 9 of \cite{Part1} and leads
to ``external field'' parameters that were introduced by Baxter \cite[section 8.12]{Baxter_book}. 
Note that the external field parameters arise for all $\g$ (they were explicitly analyzed
for $\sl_2$ in section 9.4 of \cite{Part1}), while the parameters that we have considered here involve an element of $\wedge^2\mf h$ and thus can appear only if $\g$ has rank greater than 1.

\section{RTT Relation in the Elliptic Cases}\label{elliptic}

\subsection{Construction}

In the rational and trigonometric settings, we saw that a quantum group -- the Yangian or the quantum loop group -- acts on the space of states at the end of any Wilson line.  In this section we will derive the same result in the elliptic case. In this case, what arises is Belavin's elliptic quantum group.

Let us work in the setting of \cite{Part1}, section 10.  We consider the projectively flat complex vector bundle of rank $N$ on an elliptic curve $E$ whose monodromy around the $a$ and $b$ cycles is given by matrices $A$ and $B$  that satisfy
\begin{equation}
 A^{-1} B^{-1} A B  = e^{\frac{2 \pi i }{N}} \op{Id} \;.
\end{equation} 
This defines a flat $PGL_N$ bundle on the elliptic curve, and so in particular a holomorphic $PGL_N$ bundle. 

Let us introduce  a Wilson line associated to a representation $W$, supported somewhere on $E$ and living in the $x$ direction, which we view as horizontal. Following our analysis in sections \ref{RTT_gl} and \ref{trigonometric}, we would like to cross this with a vertical fundamental Wilson with specified incoming and outgoing states, to find operators acting on $W$.

Locally on the elliptic curve, we can lift the $PGL_N$ bundle to an $SL_N$ bundle and define a vertical Wilson line in the fundamental representation of $SL_N$.  Thus, locally on the elliptic curve, we get operators 
\begin{equation} 
T^i_j(p) : W \to W  
 \end{equation}
from putting a vertical fundamental Wilson line at $p \in E$ with incoming and outgoing states $i$ and $j$. 

These operators do not make sense globally on the elliptic curve. There is a topological obstruction to lifting the $PGL_N$ bundle to an $SL_N$ bundle,  so that the vertical Wilson line is not well-defined globally on the elliptic curve.

To understand what happens globally, we will introduce a certain covering space. We let $\pi:\til{E} \to E$ be the $N^2$ to $1$ covering space with the feature that $E$ is the quotient of $\til{E}$ by the group of points on $\til{E}$ of order $N$.  More explicitly, if $E$ is the quotient of $\C$ by the lattice spanned by $(1,\tau)$, then $\til{E}$ is the quotient of $\C$ by the lattice spanned by $(N, N \tau)$.   

Since the monodromy around the $a$ and $b$ cycles of $E$ of the flat $PGL_N$ bundle is of order $N$, the monodromy of the bundle pulled back to $\til{E}$ is trivial, and so the bundle is trivial.

Given a point $p \in E$, we can define a fundamental Wilson line at $p$ by choosing a point $z \in \til{E}$ with $\pi(z) = p$.   Once we choose such a lift, we get a trivialization of the $PGL_N$ bundle at $p$ and so a fundamental Wilson line.  Therefore, if $z \in \til{E}$, we get an operator 
\begin{equation} 
T^i_j(z) : W \to W 
 \end{equation}
from placing a vertical Wilson line at $\pi(z) \in E$.   
In perturbation theory,\footnote{This is actually not true beyond perturbation theory: the exact $R$-matrix has finitely
many poles in each fundamental domain.  To reconcile the statements,  one must bear in mind that
$1/(z-\hbar)=1/z+\hbar/z^2+\dots$ can be viewed in perturbation theory as a function that only has poles at $z=0$.}  $T^i_j(z)$ only has poles where $\pi(z)=p$.

The operator $T^i_j(z)$ depends on the choice of lift of a point $p = \pi(z)$ in $E$ to $z \in \til{E}$.  Different lifts will give rise to different trivializations of the fibre of the $PGL_N$ bundle at $p \in E$, and so to operators $T^i_j(z)$ which differ by conjugating with some matrix acting on $\C^N$.   

Choosing a basis $a,b$ for the group of order $N$ points of $\til{E}$, we have the relations
\begin{align} \label{equation_quasi_periodicity_RTT} 
\begin{split}
T(z+a) &= A T(z) A^{-1} \;,\\
T(z + b) &= B T(z) B^{-1} \;, 
\end{split}
\end{align} 
In these equations we are viewing $\{T^i_j(z)\}$ as an $N \times N$ matrix whose entries are elements of $\op{End}(W)$, and we are conjugating it with the $N \times N$ matrices $A,B$ whose entries are scalars.  These relations follow from our definition of the $PGL_N$ bundle in terms of the matrices $A$ and $B$.

One can, of course, pass to the universal cover $\C \to \til{E}$, where we view $\til{E}$ as the quotient of $\C$ by the lattice generated by $N$ and $N \tau$.  In this language, the generators of the group of $N$-torsion points on $\til{E}$ are $1$ and $\tau$, and the relations (\ref{equation_quasi_periodicity_RTT}) take the form $T(z+1) = A T(z) A^{-1}$, $T(z+\tau) = B T(z) B^{-1}$.  In this form, we are viewing $T(z)$ as a function on $\C$ with values in $\mf{gl}_N \otimes \op{End}(W)$.  Since $A^N = 1$ and $B^N= 1$, we see that $T(z+N) = T(z)$ and $T(z+N\tau) = T(z)$, so that $T(z)$ depends to the elliptic curve $\til{E}$.

Of course, the operators $T^i_j(z)$ also satisfy the RTT relation we are familiar with from the rational and trigonometric cases. In addition, $T^i_j(z)$ satisfies the quantum determinant relation  
\begin{equation}
\sum_{k_r} \op{Alt}(k_0,\dots,k_{n-1}) \, T^{k_0}_{0}(z ) \, T^{k_1}_{1}(z + 2 \hbar  )\cdots T^{k_{n-1}}_{n-1}(z + 2(n-1) \hbar ) = 1\;. \label{equation_qdet_elliptic}  
\end{equation}

It is more difficult in the elliptic case than in the rational or trigonometric cases to turn these relations on the operators $T^i_j(z)$ into the relations defining an algebra with a simple set of generators.  This is because there are no natural boundary points on the elliptic curve around which we  can expand $T^i_j(z)$ as a series in $z$.

One way around this problem, which is sometimes considered in the literature \cite{ES2008}, is to avoid defining the elliptic quantum group as an algebra, but instead to simply describe what it means to give a representation of this putative algebra.  
\begin{definition}
A representation of the elliptic quantum group for $\mf{sl}_N$ is a finite dimensional vector space $W$ with an operator $T(z) : W \to W [[\hbar]]$, which is analytic as a function of $z \in \C$ with countably many singular points. At each order in $\hbar$ every singularity is of finite order. There are finitely many singular points in each fundamental domain for the elliptic curve.  

The operator $T(z)$ must satisfy the RTT relation, the quantum determinant relation, and and the quasi-periodicity relation (\ref{equation_quasi_periodicity_RTT}).  
\end{definition}
This definition matches one considered in \cite{ES2008} (except that we have also imposed the quantum determinant relation). This definition is equivalent to (but technically more convenient than) other definitions which build the elliptic quantum group as an associative algebra. 

\subsection{\texorpdfstring{Uniqueness of the Elliptic $R$-matrix}{Uniqueness of the Elliptic R-matrix}}

In the rational and trigonometric cases, we saw that the $R$-matrix was uniquely constrained by the Yang-Baxter equation, unitarity, and certain supplementary equations. In type $A$ the supplementary equation is the quantum determinant relation.  In the elliptic case, a similar uniqueness result holds.  We will be brief since the argument is similar to that presented in the rational and trigonometric cases.

Suppose that we have two elliptic solutions $R,R'$ to the quantum Yang-Baxter equation. We will view $R,R'$ as series in $\hbar$ whose coefficients are meromorphic functions on $\C$ valued in $\mf{gl}_N \otimes \mf{gl}_N$, and which satisfy the quasi-periodicity conditions discussed above with respect to the translations $z \mapsto z + 1$, $z \mapsto z + \tau$.  We will also assume that $R,R'$ satisfy the quantum determinant relation and the unitarity condition.  

Suppose that $R,R'$ agree modulo $\hbar^k$, and that modulo $\hbar^2$ both are given by the solution $r(z) = r^{(1)}(z)$ to the classical Yang-Baxter equation that we derived from field theory in section 10 of
\cite{Part1}.   Then $R,R'$ differ by $\hbar^k r^{(k)}(z)$,  where $r^{(k)}(z)$ is a first-order deformation of $r^{(1)}(z)$ as a solution to the CYBE.  The quantum determinant relation applied to each variable tells us that $r^{(k)}(z) \in \mf{sl}_N \otimes \mf{sl}_N$.

Belavin-Drinfeld's classification  \cite{Belavin-Drinfeld}  of solutions to the CYBE in the elliptic case tells us that $r^{(k)}(z)$ is a multiple of $r^{(1)}(z)$.  Therefore $R,R'$ differ by a non-linear reparametrization of $\hbar$ of the form $\hbar \mapsto \hbar + c_2 \hbar^2 + \dots$.   

In terms of field theory, reparameterization of $\hbar$ amounts to adding counter-terms of the form $\hbar^k \int \d z \,  \textrm{CS}(A)$ at each order in the loop expansion.  Such counter-terms are not forbidden by the symmetries of the theory on $\R^2 \times E$.  We thus find that the elliptic $R$-matrix is uniquely constrained by the formal properties it satisfies, up to a reparameterization of $\hbar$ which is an inherent ambiguity in quantizing the system.

\section*{Acknowledgments}

K.~C.~is supported by the NSERC Discovery Grant program and by the Perimeter Institute for Theoretical Physics. Research at Perimeter Institute is supported by the Government of Canada through Industry Canada and by the Province of Ontario through the Ministry of Research and Innovation. 
E.~W.~is partially supported by National Science Foundation grant NSF Grant PHY-1606531.
M.~Y.~is partially supported by WPI program (MEXT, Japan), by JSPS Program No.\ R2603, by JSPS KAKENHI Grant No.\ 15K17634, and by JSPS-NRF Joint Research Project.



\bibliographystyle{unsrt}

\begin{thebibliography}{99}

\bibitem{Costello_2013}
  K.~Costello,
  ``Supersymmetric Gauge Theory and the Yangian,''
  arXiv:1303.2632 [hep-th].

\bibitem{Part1} 
  K.~Costello, E.~Witten and M.~Yamazaki,
  ``Gauge Theory and Integrability, I,''\\
  arXiv:1709.09993 [hep-th].

\bibitem{Witten_2016}
  E.~Witten,
  ``Integrable Lattice Models From Gauge Theory,''
  arXiv:1611.00592 [hep-th].
  
\bibitem{Takhtajan-Faddeev} 
  L.~A.~Takhtajan and L.~D.~Faddeev,
  ``The Quantum method of the inverse problem and the Heisenberg XYZ model,''
  Russ.\ Math.\ Surveys {\bf 34}, no. 5, 11 (1979)
  [Usp.\ Mat.\ Nauk {\bf 34}, no. 5, 13 (1979)].
  
\bibitem{Kulish-Sklyanin} 
  P.~P.~Kulish and E.~K.~Sklyanin,
  ``On the solution of the Yang-Baxter equation,''
  J.\ Sov.\ Math.\  {\bf 19}, 1596 (1982)
  [Zap.\ Nauchn.\ Semin.\  {\bf 95}, 129 (1980)].
  
\bibitem{Drinfeld_Hopf}
 V.~G.~Drinfeld, ``Hopf Algebras and the Quantum Yang-Baxter Equation,''
  Sov.\ Math.\ Dokl.\  {\bf 32}, 254 (1985)
  [Dokl.\ Akad.\ Nauk Ser.\ Fiz.\  {\bf 283}, 1060 (1985)].
   
\bibitem{Drinfeld_original} 
  V.~G.~Drinfeld,
  ``Quantum Groups,''
  J.\ Sov.\ Math.\  {\bf 41}, 898 (1988)
  [Zap.\ Nauchn.\ Semin.\  {\bf 155}, 18 (1986)].

\bibitem{FRT_1988}
L.~Faddeev, N.~Reshitikhin and L.~Takhtajan, 
``Quantization of Lie Groups and Lie Algebras,''
 Algebr.\ Analiz.\ {\bf 1}, LOMI-E-87-14 (1987).    
 
 \bibitem{Chari-Pressley}
 V.~Chari and A.~Pressley, ``A Guide to Quantum Groups,'' Cambridge University Press, 1994.
  
\bibitem{Drinfeld_ICM}
 V.~G.~Drinfeld, ``Quantum Groups,'' 
 in {\it Proceedings of the International Congress of Mathematicians, Berkeley, 1986}, American Mathematical Society, 1987.

\bibitem{Ding-Frenkel}  
J.~T.~Ding and I.~B.~Frenkel, ``Isomorphism of two realizations of quantum affine algebra $U_q(\widehat{\mathfrak{gl}(n)})$,''
  Comm.\ Math.\ Phys.\ , {\bf 156}, 277 (1993).
  
\bibitem{Belavin-Drinfeld}
 A.~A.~Belavin and V.~G.~Drinfeld,
 ``Solutions of the Classical {Y}ang-{B}axter equation for Simple {L}ie Algebras,'' 
 Funktsional.\ Anal.\ i Prilozhen.\  {\bf 16}, 1 (1982).  

%
  
\bibitem{Baxter_book}
R.~J.~Baxter, ``Exactly Solved Models in Statistical Mechanics,'' Academic Press, 1982.  

\bibitem{ES2008}
E.~Etingof and O.~Schiffmann,
``A Link Between Two Elliptic Quantum Groups,'' \\
arXiv:math/9801108 [math.QA].

\end{thebibliography}

\end{document}